\documentclass[10pt,a4paper]{article}
\usepackage{amssymb}
\usepackage{amsmath}
\usepackage{amsthm}
\usepackage{latexsym}
\usepackage[dvips]{epsfig}
\usepackage{enumerate}
\usepackage{mathrsfs}
\usepackage{eufrak}
\usepackage{bm}
\usepackage{tikz}
\usepackage{authblk}
\usepackage{cancel}
\usepackage{framed}

\usepackage{xr}
\theoremstyle{plain}
\newtheorem{proposition}{Proposition}
\newtheorem{lemma}{Lemma}
\newtheorem{theorem}{Theorem}
\newtheorem{assumption}{Assumption}

\newtheorem{remark}{Remark}
\newtheorem{gauge}{Gauge choice}

\def\p{\partial}

\def\bmd{{\bm d}}
\def\bme{{\bm e}}

\def\bmg{{\bm g}}
\def\bmh{{\bm h}}

\def\bml{{\bm l}}
\def\bmn{{\bm n}}
\def\bmm{{\bm m}}

\def\bmu{{\bm u}}

\def\bmB{{\bm B}}

\def\bmeta{{\bm \eta}}

\def\bmvarphi{{\bm \varphi}}
\def\bmsigma{{\bm \sigma}}

\def\bmpartial{{\bm \partial}}

\def\nablasl{/\kern-0.58em\nabla}
\def\Deltasl{/\kern-0.58em\Delta}

\allowdisplaybreaks

\begin{document}

\title{\textbf{The conformal Einstein field equations and the local
    extension of future null infinity}}

\author[,1,3]{Peng
  Zhao \footnote{E-mail address:{\tt p.zhao@qmul.ac.uk}}}
\author[,2]{David Hilditch \footnote{E-mail address:{\tt
      david.hilditch@tecnico.ulisboa.pt}}}
\author[,3]{Juan A. Valiente
    Kroon \footnote{E-mail address:{\tt
        j.a.valiente-kroon@qmul.ac.uk}}}
\affil[1]{College of Education for the Future, Beijing Normal University
 at Zhuhai, No.18, Jinfeng Road, Tangjiawan, Zhuhai City, Guangdong
  Province, 519087, P.R.China.}        
\affil[2]{CENTRA, Departamento de F\'isica, Instituto Superior
  T\'ecnico – IST, Universidade de Lisboa – UL, Avenida Rovisco Pais
  1, 1049 Lisboa, Portugal.}
\affil[3]{School of Mathematical Sciences, Queen Mary, University of
    London, Mile End Road, London E1 4NS, United Kingdom.}

\maketitle

\begin{abstract}
We make use of an improved existence result for the characteristic
initial value problem for the conformal Einstein equations to show
that given initial data on two null hypersurfaces~$\mathcal{N}_\star$
and~$\mathcal{N}'_\star$ such that the conformal factor (but not its
gradient) vanishes on a section of~$\mathcal{N}_\star$ one recovers a
portion of null infinity. This result combined with the theory of the
hyperboloidal initial value problem for the conformal Einstein field
equations allows to show the semi-global stability of the Minkowski
spacetime from characteristic initial data.
\end{abstract}

\section{Introduction}

This is the third article in a series devoted to the analysis of the
characteristic initial value problem (CIVP) for the Einstein field
equations. This research programme is motivated by the techniques
introduced by Luk in~\cite{Luk12} to obtain local existence results
for the Einstein field equations which are optimal in the sense that
one obtains a solution in a neighbourhood of both the initial null
hypersurfaces and not only in a neighbourhood of their intersection as
in Rendall's original approach~\cite{Ren90}. In Paper~I of this
series, see~\cite{HilValZha19}, we obtained an improved local
existence result for the CIVP for the Einstein field equations
expressed in terms of the Newman-Penrose formalism and a gauge due to
Stewart ---see~\cite{Ste91}. This result demonstrates the robustness
of Luk's approach, showing that the specific choice of gauge employed
in~\cite{Luk12} is not crucial. In Paper~II of this series,
see~\cite{HilValZha20}, we applied Luk's method to obtain a local
existence result for the asymptotic CIVP for Friedrich's conformal
Einstein field equations. In this problem, one of the initial
hypersurfaces is past null infinity while the other is an incoming
light cone ---in an alternative version of this problem one prescribes
data on future null infinity and an outgoing null hypersurface.

\smallskip
\noindent
\textbf{The problem.} In this article we make use of a CIVP for the
conformal Einstein field equations to study the question of \emph{the
  local extendibility of null infinity}. To this end, initial data is
prescribed on two future oriented null hypersurfaces intersecting a
2-dimensional surface with the topology of the
2-sphere~$\mathbb{S}^2$. These null hypersurfaces are assumed to
intersect future null infinity, $\mathscr{I}^+$. The question to be
addressed is whether it is possible to recover a portion of
future null infinity lying in the causal future of the initial
hypersurfaces. Observe that in the future null infinity version of the
asymptotic CIVP analysed in Paper~II, the solution constructed is
located in the causal past of the initial hypersurfaces ---see
Figure~\ref{Fig:AsymptoticCharacteristicProblems}. The question of the
local extendibility of null infinity through a CIVP has been studied
by Li \& Zhu in~\cite{LiZhu18} directly through the Einstein field
equations. In this work, in order to encode the asymptotic
behaviour of the various field at infinity it is necessary to make use
of weighted function spaces and norms. Moreover, it is necessary to
consider the existence of solutions to the field equations on a domain
with an infinite extent. In this article we make use of an alternative
setup that offers a natural way to address the local extendibility of
null infinity: \emph{the use of a conformal representation of the
  spacetime and the conformal Einstein field equations} ---see
e.g.~\cite{CFEBook}.

\smallskip
\noindent
\textbf{Conformal methods.} The use of conformal methods in the study
of the local extendibility of (future) null infinity allows to
transform the question of existence of solutions to hyperbolic
evolution equations on an infinite domain into the study of solutions
on a finite region. Moreover, the asymptotic decay of the various
fields fields is conveniently encoded through regularity of the
fields. Accordingly, it is possible to work with standard (unweighted)
function spaces and norms. Luk's strategy to analyse the CIVP allows
to ensure the existence of solutions on \emph{causal diamonds} having
a \emph{long} and a \emph{short} direction ---see
Figure~\ref{Fig:AsymptoticCharacteristicProblems}. Existence in the
long direction is ensured as long as one has control on the initial
data. On the other hand, the extent of the short direction is
restricted by the potential appearance of \emph{singularities in
  finite time} due to the presence of \emph{Riccati-type} equations in
the evolution system. In the present problem the conformal framework
provides a natural causal diamond with one of its sides lying on one
of the null initial hypersurfaces, $\mathcal{N}_\star$, and a short
side covering a portion of null infinity. Although from the point of
view of the conformal representation this domain has a finite size, in
the physical spacetime it actually represents an infinite domain
contained between two \emph{parallel} null hypersurfaces. The main
result of this article is that \emph{it is possible to ensure the
  existence of solutions to the conformal Einstein field equations on
  the causal diamond with sides on~$\mathcal{N}_\star$
  and~$\mathscr{I}^+$}. Thus, it is possible to recover a portion of
null infinity to the future of $\mathcal{N}_\star$ ---i.e. we have
extended~$\mathscr{I}^+$.

\begin{figure}[t]
\centering
\includegraphics[width=\textwidth]{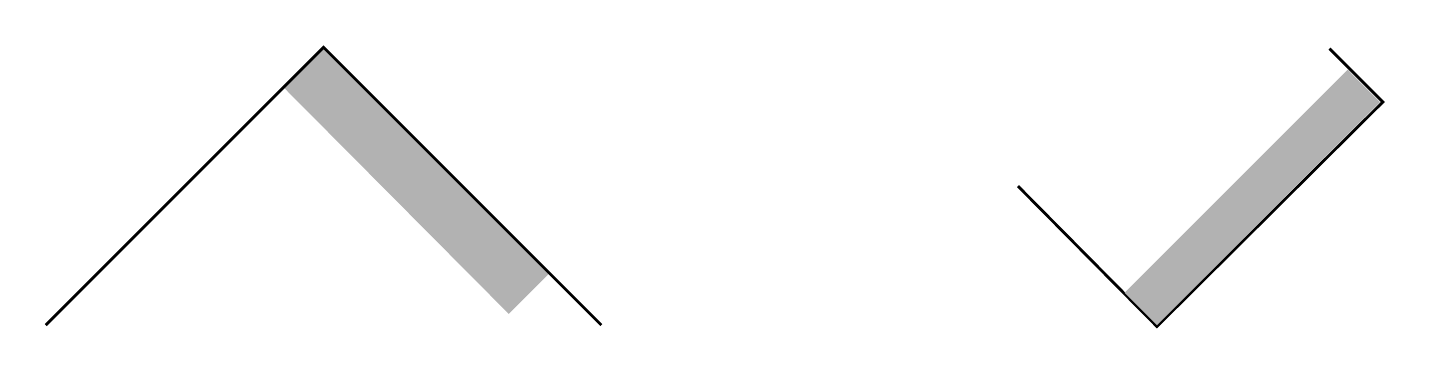}
\put(-420,100){(a)} \put(-140,100){(b)} \put(-300,80){$\mathscr{I}^+$}
\put(-15,90){$\mathscr{I}^+$} \put(-392,50){$\mathcal{N}_\star$}
\put(-35,50){$\mathcal{N}_\star$} \put(-120,30){$\mathcal{N}_\star'$}
\caption{Comparison between the asymptotic characteristic problem (a)
  and the standard characteristic problem (b) for the conformal
  Einstein field equations. In the future null infinity version of the
  asymptotic CIVP initial data is prescribed on future null infinity
  and on an outgoing lightcone~$\mathcal{N}_\star$. The neighbourhood
  theorem allows the recovery of a narrow causal diamond along null
  infinity. The length of this rectangle is limited by the portion
  of~$\mathscr{I}^+$ on which one has control of the initial
  data. Observe that region of existence of solutions lies in the
  causal past of the null hypersurfaces and that the existence of, at
  least a portion of null infinity is \emph{a priori} assumed. In the
  characteristic problem considered in this article the initial data
  is prescribed on two standard null hypersurfaces~$\mathcal{N}_\star$
  and~$\mathcal{N}_\star'$ with at least one of
  them~($\mathcal{N}_\star$) intersecting the conformal boundary. The
  improved existence result allows then to recover a narrow rectangle
  whose long side lies on~$\mathcal{N}_\star$ and the short one gives
  a portion of future null infinity. Observe that the region of
  existence is on the causal future of the initial hypersurfaces and
  that, \emph{a priori} only the existence of a cut of null infinity
  is assumed.
\label{Fig:AsymptoticCharacteristicProblems}}
\end{figure}
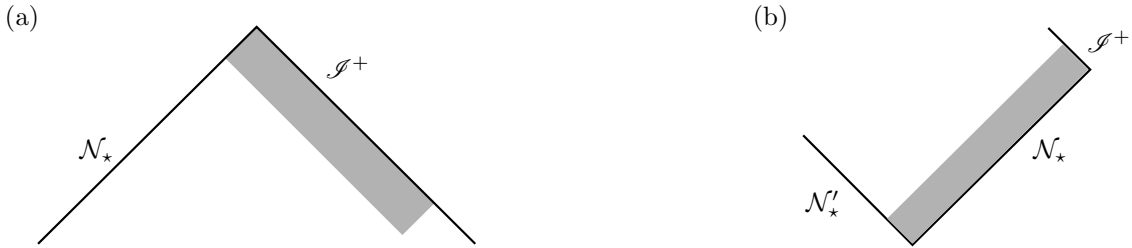

\smallskip
\noindent
\textbf{The hyperboloidal initial value problem.} Historically, the
first resolution of the local extendibility of null infinity has been
given by Friedrich in his analysis of the hyperboloidal initial value
problem for the conformal Einstein field equations
---see~\cite{Fri84,Fri86b}, also~\cite{LueVal09,CFEBook}. In this
case, initial data is prescribed on a spatial
hypersurface~$\mathcal{H}_\star$ which intersects null infinity. Due
to the formal regularity of the conformal Einstein field equations at
the conformal boundary, the standard local existence theory for
symmetric hyperbolic systems allows to recover a slab of spacetime in
the causal future of~$\mathcal{H}_\star$ which covers a portion of
null infinity ---see Figure~\ref{Fig:HyperboloidalProblem}. A
particular drawback of this approach, in contrast with the CIVP, is
the increased complexity in solving the constraint equations
on~$\mathcal{H}_\star$ and obtaining conditions ensuring peeling (see
below) ---see~\cite{AndChrFri92,AndChr93,AndChr94}. In view of the
latter and the historical and practical relevance of the CIVP it is of
interest to discuss the extendibility of null infinity form this
alternative point of view.

\begin{figure}[t]
\centering
\includegraphics[width=0.3\textwidth]{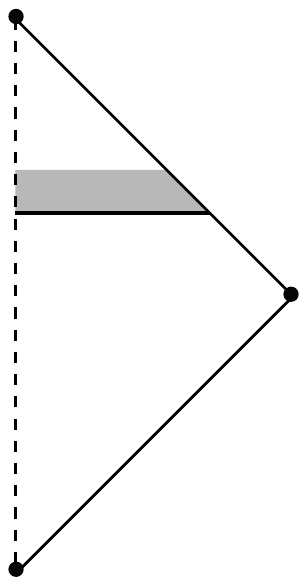}
\put(-50,140){$\mathscr{I}^+$}
\put(-70,105){$\mathcal{H}_\star$}
\caption{A schematic depiction of a hyperboloidal IVP in an
    asymptotically flat spacetime. Suitable initial data is given
    on~$\mathcal{H}_\star$, and local existence is guaranteed in the
    (local) shaded region provided regularity and symmetric
    hyperbolicity of the equations of motion under consideration.
\label{Fig:HyperboloidalProblem}}
\end{figure}

\smallskip
\noindent
\textbf{Peeling.} The problem here considered is closely related to
one of the central issues on the study of the asymptotics of the
gravitational field: \emph{peeling}. As part of the formulation of the
CIVP here considered it is necessary to prescribe the value of one of
the components of the Weyl tensor~($\phi_0$) on one of the null
hypersurfaces. This component is usually loosely interpreted as
describing some sort of incoming radiation ---see~\cite{Sze65}. For
simplicity, in the present analysis it is assumed that \emph{the
  component~$\phi_0$ is smooth at null infinity}. It follows that on
the portion of future null infinity recovered by the optimal local
existence result for the CIVP \emph{the Weyl tensor satisfies the
  peeling behaviour}. If a finite regularity is assumed below a
certain threshold, then the assumptions of the \emph{peeling theorem}
are no longer satisfied ---see e.g.~\cite{PenRin86,Ste91,CFEBook}.

\smallskip
\noindent
\textbf{Differences with the asymptotic characteristic problem}. The
CIVP considered in this article differs from that in Paper~II in that
in the former one of the initial hypersurfaces coincides with the
conformal boundary. This leads to a number of simplifications in the
gauge and equations. In the present case, both initial null
hypersurfaces lie in the physical spacetime ---except for their
intersections with null infinity. Thus, one has to deal with a
somewhat more general set up. Nevertheless, a careful inspection of
the analysis of Paper~II shows that all the main assertions and
estimates hold in the present situation. Roughly speaking these
estimates control the size of the~$L^2$-norm of the fields appearing
in the conformal Einstein field equations in terms of the size of the
initial data. Thus, if the data is finite, so will also the solutions
to the conformal Einstein field equations. The existence of solutions
on the causal diamond containing a portion of null infinity then
follows from a \emph{last slice argument} in which the basic existence
domain arising from the use of \emph{Rendall's reduction
  strategy}~\cite{Ren90} is progressively extended.

\smallskip
\noindent
\textbf{An application: the semi-global stability of Minkowski
  spacetime from a Cauchy-characteristic problem.} As an application
of our result on the local extendibility of null infinity in the CIVP
for the conformal Einstein field equations, we obtain a semi-global
stability result for the Minkowski spacetime. The idea behind this
construction is the following: given standard Cauchy initial data for
the conformal Einstein field equations on a compact spacelike
domain~$\mathcal{K}_\star$, and characteristic data up to the
conformal boundary on an outgoing null
hypersurface~$\mathcal{N}_\star$ emanating from the boundary of the
spacelike domain the development of Cauchy data implies complementary
characteristic data on the Cauchy horizon~$H^+(\mathcal{K}_\star)$ of
the development. In turn, setting~$\mathcal{N}_\star
=H^+(\mathcal{K}_\star)$, the local extendibility result of null
infinity can be used to obtain two (intersecting) causal diamonds
along the initial null hypersurfaces~$\mathcal{N}_\star$
and~$\mathcal{N}_\star'$ both including a portion of future null
infinity. Accordingly, the existence domain along the initial
hypersurface~$\mathcal{K}_\star\cup\mathcal{N}_\star$ will contain a
hyperboloidal hypersurface~$\mathcal{H}_\star$. If the initial data
on~$\mathcal{K}_\star\cup\mathcal{N}_\star$ is assumed to be suitably
close to data for the Minkowski spacetime, then the initial data
induced on~$\mathcal{H}_\star$ will also be close to Minkowski
hyperboloidal data. One can then use Friedrich's semi-global existence
stability results in~\cite{Fri86b} ---see also~\cite{LueVal09}--- to
recover the whole of the domain of
dependence~$D^+(\mathcal{K}_\star\cup \mathcal{N}_\star)$. For a
recent alternative proof of the stability of the Minkowski spacetime
from Cauchy-characteristic data which makes use of the full machinery
of vector field methods ---see~\cite{Gra20}.

\subsection*{Conventions}

This article follows the conventions and notation of Paper~I and
Paper~II, which, in turn follow those used in the
monograph~\cite{CFEBook}. The latter reference is consistent with the
presentation in J. Stewart's book~\cite{Ste91} and Penrose \&
Rindler~\cite{PenRin84,PenRin86}.

\section{The conformal vacuum Einstein field equations}

The main technical tool for the analysis of the the local extension of
future null infinity are Friedrich's \emph{conformal vacuum Einstein
  field equations (CEFE)}. The equations are a conformal
representation of the vacuum Einstein field equations. Crucially, they
are \emph{formally regular} on the conformal boundary and imply, away
from it, a solution to the vacuum Einstein field equations. The
structural properties of the CEFE and its derivation have been amply
discussed in the literature ---see~\cite{Fri84,CFEBook}.

\smallskip
In what follows, let~($\mathcal{M}, \bmg, \Xi$) denote a conformal
extension of a vacuum \emph{asymptotically simple spacetime}
(see~\cite{PenRin86,CFEBook} for a definition)~$(\tilde{\mathcal{M}}$,
$\tilde{\bmg}$. The \emph{physical metric}~$\tilde\bmg$ and the
\emph{unphysical metric}~$\bmg$ are related to each other via the
formula~$\bmg=\Xi^2\tilde{\bmg}$. By assumption, the unphysical
manifold~$\mathcal{M}$ has a boundary ---the \emph{conformal
  boundary}, $\mathscr{I}\equiv\p\mathcal{M}$, corresponding to
the future endpoints of null geodesics. The conformal factor
satisfies~$\Xi>0$ on $\mathcal{M}\setminus\mathscr{I}$
and~$\Xi=0$,~$\bmd\Xi\neq0$ on $\mathscr{I}$ ---that is,~$\Xi$ is a
boundary defining function.

\smallskip
The metric vacuum conformal Einstein field equations with vanishing
Cosmological constant are given by the system
\begin{subequations}
\begin{align}
&\nabla_a\nabla_b\Xi=-\Xi L_{ab}+sg_{ab}, \label{CFE1}\\
&\nabla_as=-L_{ac}\nabla^c\Xi, \label{CFE2}\\
&\nabla_cL_{db}-\nabla_dL_{cb}=\nabla_a\Xi d^a{}_{bcd}, \label{CFE3}\\
&\nabla_a d^a{}_{bcd}=0, \label{CFE4}\\
&6\Xi s-3\nabla_c\Xi\nabla^c\Xi=0 \label{CFE5},\\
&R^c{}_{dab} = C^c{}_{dab} + 2 \big(\delta^c{}_{[a} L_{b]d} - g_{d[a}
   L_{b]}{}^c\big), \label{CFE6}
\end{align}
\end{subequations}
where
\begin{align*}
L_{ab}\equiv\frac{1}{2}R_{ab}-\frac{1}{12}Rg_{ab},\qquad
d^a_{\phantom{a}bcd}\equiv\Xi^{-1}C^a_{\phantom{a}bcd},\qquad
s\equiv\frac{1}{4}\nabla^a\nabla_a\Xi+\frac{1}{24}R\Xi,
\end{align*}
are, respectively, the \emph{Schouten tensor}, the \emph{rescaled Weyl
  tensor} and the \emph{Friedrich scalar}. For convenience we also
define
\begin{align*}
\Sigma_a \equiv \nabla_a \Xi.
\end{align*}

The analysis of this article will be carried out with a
\emph{Newman-Penrose (NP)} version of the above equations in which the
various tensor field and equations are expressed in terms of a null
(NP) tetrad. The detailed form of these equations can be found in the
Appendix of Paper~II.

\section{The geometry of the problem}
\label{Section:GeometryProblem}

In this section we discuss the geometric setting of the local
extension of future null infinity. This is very similar to the one
used in Papers~I and~II and makes use of a gauge which we will call
\emph{Stewart's gauge}. The reader is referred
to~\cite{HilValZha19,HilValZha20} for further details and discussion
---see also~\cite{SteFri82,Ste91}.

\begin{figure}[t]
\centering
\includegraphics[width=0.8\textwidth]{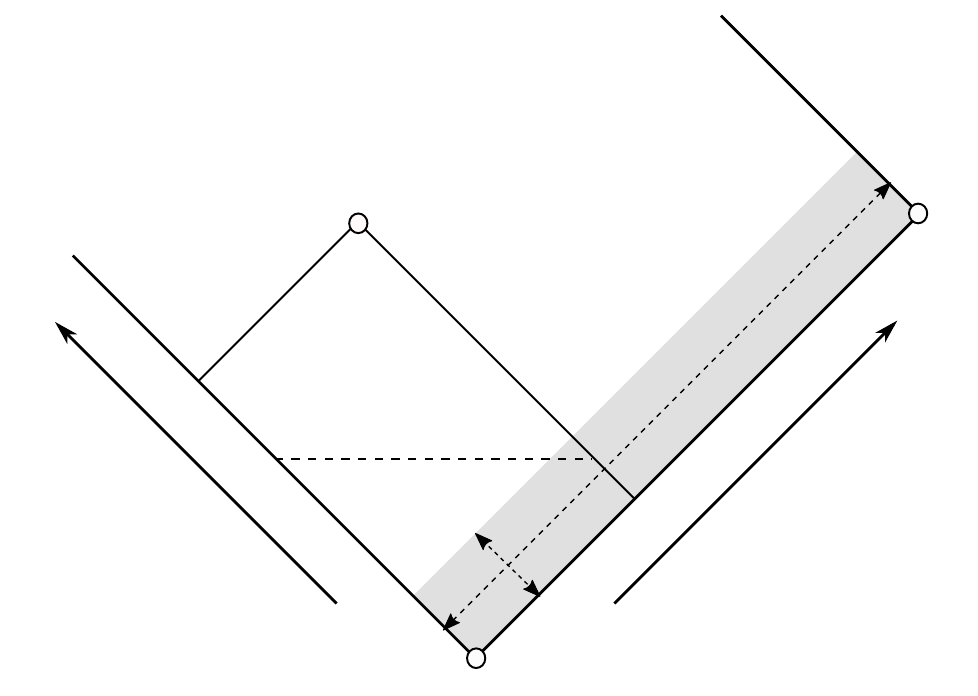}
\put(-85,60){$l^a$, $v$}
\put(-277,60){$n^a$, $u$}
\put(-60,210){$\mathscr{I}^+$}
\put(-315,130){$\mathcal{N}'_\star$}
\put(-50,130){$\mathcal{N}_\star$}
\put(-15,164){$\mathcal{S}'_\star$}
\put(-175,0){$\mathcal{S}_\star=\mathcal{S}_{u_\star,v_\star}$}
\put(-200,160){$\mathcal{S}_{u,v}$}
\put(-200,100){$\mathcal{D}_{u,v}$}
\put(-200,60){$\mathcal{D}^t_{u,v}$}
\put(-142,73){$t$}
\put(-156,41){$\varepsilon$}
\put(-267,125){$\mathcal{N}_u$}
\put(-173,125){$\mathcal{N}'_v$}
\caption{Geometric setup for the analysis of the local extension of
  future null infinity. The construction makes use of a double null
  foliation of the domain of dependence of the initial
  hypersurface~$\mathcal{N}'_\star \cup\mathcal{N}_\star$. The null
  hypersurface~$\mathcal{N}_\star$ terminates at the conformal
  boundary where~$\Xi=0$. Our construction allows us to recover a
  portion of length~$\epsilon$ on~$\mathscr{I}^+$. The coordinates and
  null NP tetrad are adapted to this geometric setting. The analysis
  is focused on the thing grey rectangular domain
  along~$\mathcal{N}_\star$. The conformal Einstein field equation
  allows to treat this problem on an infinite domain in terms of a
  problem in unphysical space on a finite
  domain. \label{Fig:CharacteristicSetup}}
\end{figure}
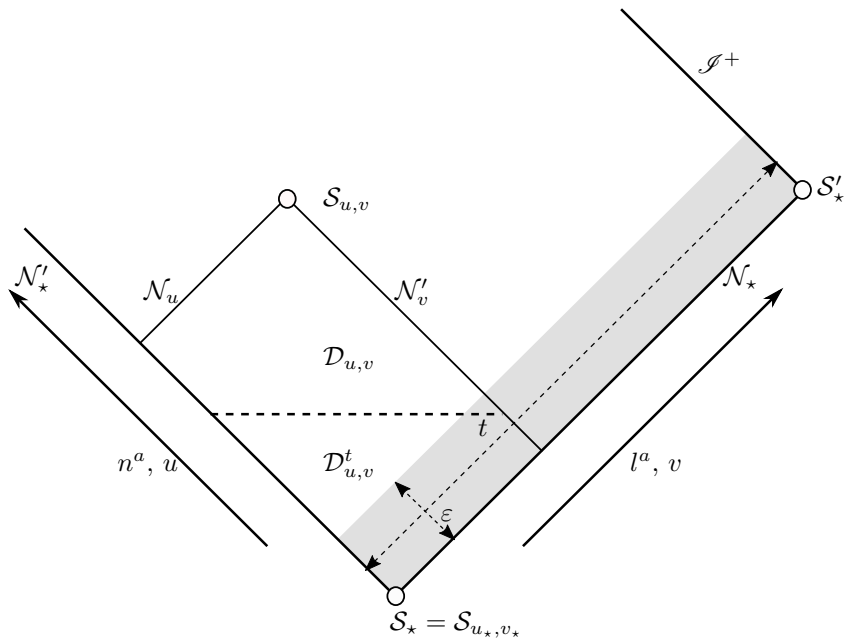

\subsection{Basic geometric setting}

In our basic setting, the unphysical manifold~$\mathcal{M}$ has a
boundary and two edges.  The boundary consists of three null
hypersurfaces: the outgoing null hypersurface~$\mathcal{N}_{\star}$;
the incoming null hypersurface~$\mathcal{N}_{\star}^{\prime}$ with
non-vacuum intersection~$\mathcal{S}_{\star}\equiv
\mathcal{N}_{\star}\cap\mathcal{N}_{\star}^{\prime}$; future null
infinity $\mathscr{I}^+$ intersecting with~$\mathcal{N}_{\star}$ at
the corner $\mathcal{S}'_{\star}$. For concreteness, we will assume
that~$\mathcal{S}_{\star},\mathcal{S}'_{\star}\approx\mathbb{S}^2$. See
Figure~\ref{Fig:CharacteristicSetup} for further details.

One can introduce coordinates~$\bar{x}=(x^{\mu})$ in a
neighbourhood~$\mathcal{U}$ of~$\mathcal{Z}_{\star}$ with~$x^0=v$
and~$x^1=u$ such that, at least in a neighbourhood
of~$\mathcal{S}_{\star}$ one can write
\begin{align*}
\mathcal{N}_{\star}=\{p\in\mathcal{U}\mid u(p)=0\}, \qquad
\mathcal{N}_{\star}^{\prime}=\{p\in\mathcal{U}\mid v(p)=0\}.
\end{align*}
Given suitable data
on~$(\mathcal{N}_{\star}\cap\mathcal{N}_{\star}^{\prime})\cap\mathcal{U}$
we are interested in making statements about the existence and
uniqueness of solutions to the CEFE on some open set
\begin{align*}
\mathcal{V}\subset\{p\in\mathcal{U}\mid u(p)\geq0, v(p)\geq0\}
\end{align*}
which we identify with a subset of the future domain of
dependence,~$D^{+}(\mathcal{N}_{\star}\cup\mathcal{N}_{\star}^{\prime})$
of~$\mathcal{N}_{\star}\cup\mathcal{N}_{\star}^{\prime}$. Moreover, we
want to show that the existence region can be extended
along~$\mathcal{N}_\star$ to reach the conformal conformal boundary
---this \emph{improved existence domain} corresponds to the grey
rectangle in Figure~\ref{Fig:CharacteristicSetup}.

\subsection{Stewart's Gauge}
\label{Section:StewartGauge}

Following the discussion of Papers~I and~II, in the following we
assume that the future of~$\mathcal{S}_{\star}$ can be foliated by a
family of null hypersurfaces:~$\mathcal{N}_{u}$ (the \emph{outgoing
  null hypersurfaces}) and~$\mathcal{N}_{v}^{\prime}$ (the
\emph{ingoing null hypersurfaces}). The scalars~$u$ and~$v$ each
satisfy the eikonal equation
\begin{align*}
\bmg^{\sharp}(\mathbf{d}u,\mathbf{d}u)=\bmg^{\sharp}(\mathbf{d}v,\mathbf{d}v)
=0.
\end{align*}
In particular, we assume that~$\mathcal{N}_{0}=\mathcal{N}_{\star}$
and~$\mathcal{N}_{0}^{\prime}=\mathcal{N}_{\star}^{\prime}$. Following
standard usage, we call $u$ a \emph{retarded time} and~$v$ an
\emph{advanced time} and use these two scalar fields~$u$ and~$v$ as
coordinates in a neighbourhood of~$\mathcal{S}_{\star}$. To complete
the coordinate system, consider arbitrary
coordinates~$(x^{\mathcal{A}})$ on~$\mathcal{D}_{\star}$, with the
index $^{\mathcal{A}}$ taking the values $2,\, 3$. These coordinates
are then propagated into~$\mathcal{N}_{\star}$ by requiring them to be
constant along the generators of~$\mathcal{N}_{\star}$. Once
coordinates have been defined on~$\mathcal{N}_{\star}$, one can
propagate them into~$\mathcal{V}$ by requiring them to be constant
along the generators of each~$\mathcal{N}_{v}^{\prime}$. In this
manner one obtains a coordinate system~$(x^{\mu})=(u,\, v,\,
x^{\mathcal{A}})$ in $\mathcal{V}$. Moreover, we
define~$\mathcal{S}_{u,v} \equiv \mathcal{N}_u\cap
\mathcal{N}_{v}^{\prime} \approx \mathbb{S}^2$. Our analysis will be
mostly carried out in \emph{causal diamonds} of the form
\begin{align*}
\mathcal{D}_{u',v'} \equiv \{ 0\leq v \leq v', \; 0\leq u \leq u' \} =
\cup_{0\leq v \leq v',  0\leq u \leq u'} \mathcal{S}_{u,v}.
\end{align*}
By means of the time function $t\equiv u+v$ one can readily define the
\emph{truncated causal diamond}
\begin{align*}
\mathcal{D}^{\tilde{t}}_{u',v'} \equiv \mathcal{D}_{u',v'} \cap \{
t\leq \tilde{t}\}.
\end{align*}

\smallskip
The above coordinate construction is complemented by an NP null
tetrad~$\{\bml,\, \bmn,\, \bmm,\, \bar\bmm\}$ with the vectors~$\bml$
and~$\bmn$ tangent to the generators of the null
hypersurfaces~$\mathcal{N}_{u}$ and~$\mathcal{N}_{v}^{\prime}$
respectively. Following the same discussion of Papers~I and~II we make

\begin{gauge}[\textbf{\em Stewart's choice of the components of the
      frame}]\label{Assumption:Stewarts_Frame} {\em On~$\mathcal{V}$
    we consider a NP frame of the form
\begin{align}
\bml=\bmpartial_v +C^{\mathcal{A}}\bmpartial_{\mathcal{A}},\ \
\bmn=Q\bmpartial_u, \ \
\bmm=P^{\mathcal{A}} \bmpartial_{\mathcal{A}}, \label{framem}
  \end{align}}
where~$C^{\mathcal{A}}=0$ on~$\mathcal{N}_{\star}$, $\bmm$
and~$\bar\bmm$ span the tangent space of
$\mathcal{S}_{u,v}$. On~$\mathcal{N}_{\star}^{\prime}$ one has
that~$\bmn=Q\bmpartial_u$. As the coordinates~$(x^{\mathcal{A}})$ are
constant along the generators of~$\mathcal{N}_{\star}$
and~$\mathcal{N}_{\star}^{\prime}$, it follows that
on~$\mathcal{N}_{\star}^{\prime}$ the coefficient~$Q$ is only a
function of~$u$. Thus, without loss of generality one can
re-parameterise~$u$ so as to set~$Q=1$
on~$\mathcal{N}_{\star}^{\prime}$.
\end{gauge}

Direct inspection of the NP commutators applied to the
coordinates~$(u,\, v,\, x^{\mathcal{A}})$ leads to the following:
\begin{lemma}[\textbf{\em conditions on the connection coefficients}]
\label{Lemma1}
The NP frame of the Gauge Choice~\ref{Assumption:Stewarts_Frame} can
be chosen such that
\begin{subequations}
\begin{align}
& \kappa=\nu=\gamma=0, \label{spinconnection1}\\
& \rho=\bar{\rho},\ \ \mu=\bar{\mu}, \label{spinconnection2}\\
& \pi=\alpha+\bar{\beta}, \label{spinconnection3}
\end{align}
\end{subequations}
on~$\mathcal{V}$ and, furthermore, with
\begin{align*}
\epsilon-\bar{\epsilon}=0\ \ \ \textrm{on}\ \ \
\mathcal{V}\cap\mathcal{N}_{\star}. 
\end{align*}
\end{lemma}

\begin{remark}
{\em Additional commutator relations can be used to obtain equations
  for the frame coefficients~$Q$, $P^{\mathcal{A}}$
  and~$C^{\mathcal{A}}$ ---see
  equations~\eqref{PaperII-framecoefficient1}-\eqref{PaperII-framecoefficient6}
  in Paper II.}
\end{remark}

In addition to the \emph{coordinate} and \emph{frame gauge freedom} we
also need to fix the \emph{conformal gauge freedom}. This is done in
the following lemma whose proof follows the same scheme as
Lemma~\ref{PaperII-Lemma:ConformalGauge} in Paper~II:

\begin{lemma}[\textbf{\em conformal gauge conditions for
    characteristic problem}] \label{Lemma2}
Let~$(\tilde{\mathcal{M}},\tilde{\bmg})$ denote a vacuum
asymptotically simple spacetime and let~$(\mathcal{M}, \bmg, \Xi)$
with~$\bmg=\Xi^2\tilde{\bmg}$ a conformal extension. Given the NP
frame of the Gauge Choice~\ref{Assumption:Stewarts_Frame}, the
conformal factor~$\Xi$ can be chosen so that
\begin{align*}
\mathit{R}[\bmg]=R(x), \qquad \textrm{in a neighbourhood} \quad
\mathcal{V}\quad  \textrm{of} \quad \mathcal{S}_{\star} \quad \textrm{on}
\quad J^+(\mathcal{S}_{\star})
\end{align*}
where~$R(x)$ is an arbitrary function of the coordinates. Moreover,
one has the additional gauge conditions
\begin{align*}
& \Sigma_2=1, \ \ \ \ \textrm{on}\ \ \mathcal{S}_{\star}, \\
& \Phi_{22}=0 \ \ \ \ \textrm{on}\ \ \mathcal{N}'_{\star}, \\
& \Phi_{00}=0 \ \ \ \ \textrm{on}\ \ \mathcal{N}_{\star}.
\end{align*}
\end{lemma}

\section{The formulation of the characteristic initial value problem}

This section provides a brief discussion of the basic set up and local
existence theory of the CIVP for the conformal Einstein field
equations with data on the null hypersurfaces~$\mathcal{N}_{\star}$
and~$\mathcal{N}_{\star}^{\prime}$ using Rendall's reduction
strategy~\cite{Ren90} ---see also Section~12.5 of~\cite{CFEBook}. The
analysis is completely analogous to the one carried out in Paper~I in
which the initial value problem for the vacuum Einstein field
equations was considered ---note, by contrast, the conceptual
difference with Paper~II in which an asymptotic characteristic problem
was considered.

\subsection{Specifiable free data}

In order to obtain a solution in the
domain~$J^+(\mathcal{S}_{\star})$, we need to provide initial data for
the evolution equations
on~$\mathcal{N}_{\star}\cup\mathcal{N}_{\star}^{\prime}$. In
particular, we need to know the value of the derivatives of conformal
factor~$\{\Sigma_1,\ \Sigma_2,\ \Sigma_3,\ \Sigma_4\}$, the components
of frame~$\{C^{\mathcal{A}},\ P^{\mathcal{A}}, Q \}$, the spin
connection
coefficients~$\{\epsilon,\ \pi,\ \beta,\ \mu,\ \alpha,\ \lambda,
\ \tau,\ \sigma,\ \rho\}$, the rescaled Weyl
tensor~$\{\phi_0,\ \phi_1,\ \phi_2,\ \phi_3,\ \phi_4\}$ and the Ricci
tensor~$\{\Phi_{00},\ \Phi_{01},\ \Phi_{11},\ \Phi_{02},\ \Phi_{12},\ \Phi_{22}\}$
on the initial hypersurfaces. However, as a consequence of the
constraints implied by the CEFE, this data cannot be freely
specified. As in the case of the discussion in Papers~I and~II, The
hierarchical structure of the CEFE allows to identify the basic
reduced initial data set~$r_\star$ from which the full initial data
on~$\mathcal{N}_{\star}\cup\mathcal{N}_{\star}^{\prime}$ for the
conformal Einstein field equations can be computed. The following
lemma shows us the freely specifiable data for our characteristic
problem.

\begin{lemma}[\textbf{\em freely specifiable data for the
    characteristic problem}]\label{Lemma3} Assume that the Gauge
  Choice~\ref{Assumption:Stewarts_Frame} and the gauge conditions
  implied by Lemmas~\ref{Lemma1} and~\ref{Lemma2} are satisfied in a
  neighbourhood~$\mathcal{V}$ of~$\mathcal{S}_{\star}$. Initial data
  for the conformal Einstein field equations
  on~$\mathcal{N}_{\star}\cup\mathcal{N}_{\star}^{\prime}$ can be
  computed from the reduced data set~$\mathbf{r}_\star$ consisting of:
\begin{align*}
  &\phi_0, \;  \Xi, \quad \mbox{on}\quad \mathcal{N}_{\star},\nonumber\\
  &\phi_4\ \ on\ \ \mathcal{N}_{\star}^{\prime}, \nonumber\\
  &\lambda,\ \ \phi_2+\bar{\phi}_2,\ \ \Phi_{20},\ \ \phi_3,
  \ \ P^{\mathcal{A}},\ \ on\ \ \mathcal{S}_{\star}. \nonumber
\end{align*}
\end{lemma}

The proof of this result is completely analogous to that of
Lemma~\eqref{PaperII-Lemma3} in Paper~II ---see
also~\eqref{PaperI-Lemma:FreeDataCIVP} in Paper~I.

\smallskip
In the problem under consideration we require
that~$\mathcal{N}_{\star}$ has a finite range~$v\in[0,v_\bullet]$ and
extends to the conformal boundary~$\mathscr{I}^+$ ---i.e. future null
infinity. This idea can be encoded in the following requirements
on~$\Xi$:
\begin{align*}
&\Xi>0 \qquad \mbox{on} \quad \mathcal{N}_{\star}/\mathscr{I}^+ \\
&\Xi=0 \qquad \mbox{on} \quad \mathcal{S}_{0, v_\bullet}.
\end{align*}
In addition, it is also necessary to ensure that one remains
on~$\mathscr{I}^+$ if we move away from~$\mathcal{S}_{0, v_\bullet}$ along
the direction given by~$\bmn$. This is ensured by the following Lemma.

\begin{lemma}[\textbf{\em conditions for $\Xi$ on the conformal boundary
    $\mathscr{I}^+$}]\label{Lemma:StructureConformalBoundary}
Under the same assumptions in Lemma~\ref{Lemma3}, and with a conformal
factor satisfying
\begin{align*}
\Xi=0 \qquad \mbox{on} \quad  \ \mathcal{S}_{0,v_\bullet},
\end{align*}
we have 
\begin{align*}
\Xi=0, \quad \mathbf{d}\Xi\neq 0, \qquad
\mbox{on} \quad  \mathcal{N}_{v_\bullet}^{\prime}.
\end{align*}
\end{lemma}

\begin{proof}
From the definition of~$\Sigma_a$ and the conformal
equation~\eqref{CFE1} it follows that in our gauge one has
\begin{align*}
&\Delta\Xi=\Sigma_2, \\
&\Delta\Sigma_2=-\Xi\Phi_{22},
\end{align*}
along~$\mathcal{N}_{v_{\star}}^{\prime}$. Combining these equations we
find that
\begin{align*}
\Delta^2\Xi=-\Xi\Phi_{22},
\end{align*}
so that~$\Xi=0$ is a solution such
that~$\Xi|_{\mathcal{S}_{0,v_{\star}}}$. The theory of ordinary
differential equation shows that this is the unique solution.
\end{proof}

\subsection{The reduced conformal field equations}

In Paper~II it has been discussed how the CEFE expressed in Stewart's
gauge imply a symmetric hyperbolic evolution system. More precisely,
letting
\begin{align*}
& \bm{\Sigma}^t\equiv (\Xi, \Sigma_1,\ \Sigma_2,\ \Sigma_3,\ \Sigma_4,\
   s), \\
& \bm{e}^t\equiv (C^{\mathcal{A}},\ P^{\mathcal{A}},\ Q), \\
& \bm{\Gamma}^t\equiv (\epsilon,\ \pi,\ \beta,\ \mu,\ \alpha,\ \lambda,\
\tau,\ \sigma,\ \rho), \\
& \bm{\phi}^t\equiv (\phi_0,\ \phi_1,\ \phi_2,\ \phi_3,\ \phi_4), \\
&  \bm{\Phi}^t\equiv (\Phi_{00},\ \Phi_{01},\ \Phi_{11},\ \Phi_{02},
\Phi_{12},\ \Phi_{22}),
\end{align*}
it can be shown that 
\begin{align}
\bm{\mathcal{D}}^\mu(x,\bmu)\p_\mu\bmu=\bmB(x,\bmu)\bmu
\label{ReducedCEFE}
\end{align}
with 
\begin{align*}
\bmu=(\bm{e}^t,\bm{\Sigma}^t,\bm{\Gamma}^t,\bm{\Phi}^t,\bm{\phi}^t)^t
\end{align*}
is a symmetric hyperbolic system with respect to the direction given
by
\begin{align*}
\tau^a = l^a + n^a.
\end{align*}
In particular,~$\bm{\mathcal{D}}^\mu(x,\bmu)$ are Hermitian matrices
and~$\bmB(x,\bmu)$ are smooth matrix-valued functions of their
arguments whose explicit form will not be required in the subsequent
discussion in this section. We call the evolution
system~\eqref{ReducedCEFE} the \emph{reduced conformal Einstein field
  equations}.

\begin{remark}
{\em The propagation of the constraint equations implied by the CEFE
  on the initial
  hypersurface~$\mathcal{N}_\star\cup\mathcal{N}'_\star$ can be
  addressed along the same lines of the analysis in Section 12.5
  of~\cite{CFEBook}. It follows from the latter that a solution of the
  reduced conformal field equations on a neighbourhood~$\mathcal{V}$
  of~$\mathcal{S}_{\star}$ on~$J^+(\mathcal{S}_{\star})$ that
  coincides with initial data
  on~$\mathcal{N}_\star'\cup\mathcal{N}_{\star}$ satisfying the
  conformal equations is, in fact, a solution to the conformal
  Einstein field equations on~$\mathcal{V}$.}
\end{remark}

Rendall's approach to the existence and uniqueness of solutions of
CIVP can be obtained via an auxiliary Cauchy initial value problem on
a spacelike hypersurface~$S_\star$ denoted
by~$\{p\in\mathbb{R}\times\mathbb{R}\times\mathbb{S}^2\mid
v(p)+u(p)=0\}$. The formulation of this problem crucially depends on
Whitney's extension theorem which requires being able to evaluate all
derivatives (interior and transverse) of initial data
on~$\mathcal{N}_\star'\cup\mathcal{N}_{\star}$. A key property of the
NP equations in Stewart's gauge is that any arbitrary formal
derivatives of the unknown
functions~$\{\bm{\Sigma},\ \bm{e},\ \bm{\Gamma},\ \bm{\Phi},\ \bm{\phi}\}$
on~$\mathcal{N}_\star'\cup\mathcal{N}_{\star}$ can be computed from
the prescribed initial data~$\bm{r}_\star$ for the reduced conformal
field equations on~$\mathcal{N}_\star'\cup\mathcal{N}_{\star}$. This
observation allows to make use of Whitney's extension theorem. More
details can be found in Papers~I and~II.

Combining the previous analysis and applying the theory of CIVP for
the symmetric hyperbolic systems of Section 12.5 of~\cite{CFEBook},
one obtains the following existence result:

\begin{theorem}[\textbf{\em basic local existence and uniqueness to the standard
    asymptotic characteristic problem}] Given an smooth reduced
  initial data set~$\bm{r}_\star$
  on~$\mathcal{N}_\star'\cup\mathcal{N}_{\star}$, there exists a
  unique smooth solution to the CEFE in a neighbourhood~$\mathcal{V}$
  of~$\mathcal{Z}_{\star}$ on~$J^+(\mathcal{S}_{\star})$ which implies
  the prescribed initial data
  on~$\mathcal{N}_\star'\cup\mathcal{N}_{\star}$. Moreover, this
  solution to the conformal Einstein field equations implied, in turn,
  a solution to the vacuum Einstein field equations in a neighbourhood
  of null infinity.
\end{theorem}

\section{Basic set up for the improved existence result}

In this section we briefly review the basic technical tools used in
our construction.

\subsection{Norms}

In the following we make use the same conventions for the norms of
functions as in Paper~II ---see
Section~\ref{PaperII-Section:ImprovedSetting}.

\subsection{Estimates for the frame and the conformal factor}

The first step in the analysis of the improved existence result is to
obtain control on the coefficients of the frame and the conformal
factor. The asymptotic CIVP considered in Paper~II leads to some
non-generic simplifications which do not arise when one of the initial
null hypersurfaces is not the conformal boundary. Nevertheless, the
basic analysis follows through.

\smallskip
In the following we make use of 
\begin{align*}
  \Delta_{e_{\star},\Xi_{\star}}\equiv\max\{\sup_{\mathcal{N}_\star,\mathcal{N}'_\star}
  \left(|Q|,|Q^{-1}|,|C^{\mathcal{A}}|,|P^{\mathcal{A}}|\right),
  \sup_{\mathcal{N}_\star}(\Xi)\}
\end{align*}
to measure the size of the initial data of frame and the conformal
factor. In addition, for convenience we define the scalar
\begin{align*}
\chi\equiv \Delta\log Q,
\end{align*}
which, being a derivative of a component of the frame, is at the same
level of the connection coefficients. A direct computation using the
definition of~$\chi=\Delta\log Q$ and the NP Ricci identities yields
\begin{align}
\label{EqDchi}
D\chi=\Psi_2+\bar\Psi_2+2\alpha\tau+2\bar\beta\tau+2\bar\alpha\bar\tau
+2\beta\bar\tau+2\tau\bar\tau-(\epsilon+\bar\epsilon)\chi.
\end{align}
In view of the gauge choice~$Q=1$ on~$\mathcal{N}'_{\star}$ it follows
that~$\chi=0$ on~$\mathcal{N}'_{\star}$. We also define
\begin{align*}
\varpi \equiv \beta-\bar\alpha
\end{align*}
corresponding to the only independent component of the connection on
the spheres~$\mathcal{S}_{u,v}$.

\smallskip
In order to start the analysis we take the following:

\begin{assumption}[\textbf{\em assumption to control the coefficients
      of the frame and conformal factor}]\label{Assumption:FrameSigma}
  Assume that we have a solution to the vacuum CEFEs in Stewart's
  gauge satisfying,
\begin{align*}
  ||\{\mu, \lambda, \alpha, \beta, \tau, \chi,\Sigma_2\}
  ||_{L^{\infty}(\mathcal{S}_{u,v})}
  \leq\Delta_{\Gamma}, 
\end{align*}
on a truncated causal diamond~$\mathcal{D}_{u,v_\bullet}^{\,t}$,
where~$\Delta_{\Gamma}$ is some constant.
\end{assumption}

This assumption is initially guaranteed on a sufficiently small
diamond. With the above assumption and the definition of~$\chi$,
$\Sigma_2$ and making use of the equations for the frame coefficients
implied by the NP commutators we obtain the following basic estimates
for metric and conformal factor:

\begin{lemma}[\textbf{\em control on the metric and conformal factor}]
 Given sufficiently small~$\varepsilon>0$ there exist constants~$C_1$,
 $C_2$ and~$C_3$ depending~$\Delta_{e_{\star},\Xi_{\star}}$
 and~$\Delta_\Gamma$ such that
\begin{align*}
&||Q,Q^{-1},P^{\mathcal{A}},(P^{\mathcal{A}})^{-1}||_{L^{\infty}(\mathcal{S}_{u,v})}\leq
   C_1(\Delta_{e_{\star},\Xi_{\star}}), \\
&  ||C^{\mathcal{A}}||_{L^{\infty}(\mathcal{S}_{u,v})}
   \leq C_2(\Delta_{e_{\star},\Xi_{\star}},\Delta_{\Gamma})\varepsilon, \\
& ||\Xi||_{L^{\infty}(\mathcal{S}_{u,v})}\leq C_3(\Delta_{e_{\star},\Xi_{\star}}),
\end{align*}
on~$\mathcal{D}_{u,v_\bullet}^{\,t}$. Moreover one has
\begin{align*}
\sup_{u,v}|\mbox{\em Area}(\mathcal{S}_{u,v})
-\mbox{\em Area}(\mathcal{S}_{0,v})|\leq
C(\Delta_{e_{\star},\Xi_{\star}})\Delta_{\Gamma}\varepsilon.
\end{align*}
\end{lemma}

\section{Main analysis}

In this section we present the main analysis of the article. The
strategy followed is very similar to that in Paper~II. In view of
this, most of the proofs of the various lemmas and propositions are
omitted and we focus our attention at the points where there may be
differences in the analysis of Paper~II.

\subsection{Statement of the main result}
\label{Subsection:NormsAndMainResult}

As in Paper~II, we make use of a number of tailor-made quantities to
control the various assumptions and conclusions of the bootstrap
argument underpinning our analysis.

\begin{enumerate}[(i)]
\item Quantity controlling the initial value of the connection coefficients, given
  by
 \begin{align*}
   \Delta_{\Gamma_{\star}}
   \equiv\sup_{\mathcal{S}_{u,v}\subset\mathcal{N}_\star,\mathcal{N}'_\star}
   \sup_{\Gamma\in\{\mu,\lambda,\rho,\sigma,\alpha,\beta,\tau,\epsilon\}}
   \max\{1,\sum_{i=0}^1||
   \nablasl^i\Gamma||_{L^{\infty}(\mathcal{S}_{u,v})},\sum_{i=0}^2
   ||\nablasl^i\Gamma||_{L^4(\mathcal{S}_{u,v})},\sum_{i=0}^3
   ||\nablasl^i\Gamma||_{L^2(\mathcal{S}_{u,v})}\}.
\end{align*}

\item Quantity controlling the initial value of the derivative of conformal
  factor~$\Sigma_a$, given by
\begin{align*}
  \Delta_{\Sigma_{\star}}
  \equiv\sup_{\mathcal{S}_{u,v}\subset\mathcal{N}_\star}\sup_{j=1,...,4}
  \max\{1,\sum_{i=0}^1||\nablasl^i\Sigma_j||_{L^{\infty}(\mathcal{S}_{u,v})},\sum_{i=0}^2
  ||\nablasl^i\Sigma_j||_{L^4(\mathcal{S}_{u,v})},\sum_{i=0}^3||
  \nablasl^i\Sigma_j||_{L^2(\mathcal{S}_{u,v})}\}.
\end{align*}

\item Quantity controlling the initial value of the components of the Ricci
  curvature given by
\begin{align*}
  &\Delta_{\Phi_{\star}}\equiv
    \sup_{\mathcal{S}_{u,v}\subset\mathcal{N}_\star,\mathcal{N}'_\star}
    \sup_{\Phi\in\{\Phi_{00},\Phi_{01},\Phi_{02},\Phi_{11},\Phi_{12}\}}
    \max\{1,\sum_{i=0}^1||\nablasl^i\Phi||_{L^4(\mathcal{S}_{u,v})},
    \sum_{i=0}^2||\nablasl^i\Phi||_{L^2(\mathcal{S}_{u,v})}\} \\
  &\quad+\sum_{i=0}^3\sup_{\Phi\in\{\Phi_{00},\Phi_{01},\Phi_{02},\Phi_{11},\Phi_{12}\}}
    ||\nablasl^i\Phi||_{L^2(\mathcal{N}_\star)}
    +\sup_{\Phi\in\{\Phi_{01},\Phi_{02},\Phi_{11},\Phi_{12},\Phi_{22}\}}
    ||\nablasl^i\Phi||_{L^2(\mathcal{N}'_\star)}.
\end{align*}

\item Quantity controlling the initial value of the components of the rescaled
  Weyl curvature, given by
\begin{align*}
  &\Delta_{\phi_{\star}}\equiv
    \sup_{\mathcal{S}_{u,v}\subset\mathcal{N}_\star,\mathcal{N}'_\star}
    \sup_{\phi\in\{\phi_0,\phi_1,\phi_2,\phi_3,\phi_4\}}
    \max\{1,\sum_{i=0}^1||\nablasl^i\phi
    ||_{L^4(\mathcal{S}_{u,v})},\sum_{i=0}^2
    ||\nablasl^i\phi||_{L^2(\mathcal{S}_{u,v})}\} \\
  &\quad+\sum_{i=0}^3\sup_{\phi\in\{\phi_0,\phi_1,\phi_2,\phi_3\}}
    ||\nablasl^i\phi||_{L^2(\mathcal{N}_\star)}
    +\sup_{\phi\in\{\phi_1,\phi_2,\phi_3,\phi_4\}}
    ||\nablasl^i\phi||_{L^2(\mathcal{N}'_\star)}.
\end{align*}

\item Quantity controlling the components of the Ricci curvature
  components at later null hypersurfaces, given by
\begin{align*}
  \Delta_{\Phi}\equiv\sum_{i=0}^3
    \sup_{\Phi\in\{\Phi_{00},\Phi_{01},\Phi_{02},\Phi_{11},\Phi_{12}\}}
    ||\nablasl^i\Phi||_{L^2(\mathcal{N}_u^t)}+
    \sup_{\Phi\in\{\Phi_{01},\Phi_{02},\Phi_{11},\Phi_{12},\Phi_{22}\}}
    ||\nablasl^i\Phi||_{L^2(\mathcal{N}'_v{}^t)},
\end{align*}
where the suprema in~$u$ and~$v$ are taken
over~$\mathcal{D}^t_{u,v_\bullet}$.

\item Supremum-type norm over the~$L^2$-norm of the components of the
  Ricci curvature at spheres of constant~$u$,~$v$, given by
\begin{align*}
  \Delta_{\Phi}(\mathcal{S})\equiv\sum_{i=0}^2
    \sup_{u,v}||\nablasl^i\{\Phi_{00},\Phi_{01},\Phi_{02},
    \Phi_{11},\Phi_{12}\}||_{L^2(\mathcal{S}_{u,v})},
\end{align*}
where the supremum is taken over~$\mathcal{D}^t_{u,v_\bullet}$.

\item Norm for the components of the Weyl tensor at later null
  hypersurfaces, given by
\begin{align*}
  \Delta_{\phi}\equiv\sum_{i=0}^3
    \sup_{\phi\in\{\phi_0,\phi_1,\phi_2,\phi_3\}}
    ||\nablasl^i\phi||_{L^2(\mathcal{N}_u^t)}
    +\sup_{\phi\in\{\phi_1,\phi_2,\phi_3,\phi_4\}}
    ||\nablasl^i\phi||_{L^2(\mathcal{N}'_v{}^t)}, 
\end{align*}
where the suprema in~$u$ and~$v$ are taken
over~$\mathcal{D}^t_{u,v_\bullet}$.

\item Supremum-type norm over the~$L^2$-norm of the components of the
  rescaled Weyl curvature at spheres of constant~$u$, $v$, given by,
\begin{align*} 
  \Delta_{\phi}(\mathcal{S})\equiv\sum_{i=0}^2\sup_{u,v}
    ||\nablasl^i\{\phi_0,\phi_1,\phi_2,\phi_3\}
    ||_{L^2(\mathcal{S}_{u,v})},
\end{align*}
with the supremum taken over~$\mathcal{D}^t_{u,v_\bullet}$ and in
which~$u$ will be taken sufficiently small to apply our estimates.

\end{enumerate}

The main result of this article can be expressed, in terms of the
above quantities and norms, as:

\begin{theorem}[\textbf{\em local extension of null infinity}]
\label{Theorem:LocalExtensionNullInfinity}
Given regular initial data for the conformal Einstein field equations
on~$\mathcal{N}_{\star}\cup\mathcal{N}'_\star$ such that
$\Xi|_{v=v_\bullet}$ for some $v_\bullet\in[0,\infty)$, there exists
  $\varepsilon>0$ such that an unique smooth solution to the vacuum
  conformal Einstein field equations exists in the region
\begin{align*}
\mathcal{D} \equiv \{ 0\leq u \leq \varepsilon, \; 0\leq v \leq v_\bullet\}
\end{align*}
and such that~$\varepsilon_{\star}$ can be chosen to depend only
on~$\Delta_{e_{\star},\Xi_{\star}}$, $\Delta_{\Gamma_{\star}}$,
$\Delta_{\Sigma_{\star}}$, $\Delta_{\Phi_{\star}}$
and~$\Delta_{\phi_{\star}}$. The set defined by the
condition~$v=v_\bullet$ can be identified with a portion of future
null infinity~$\mathscr{I}^+$. Furthermore, on~$\mathcal{D}$ one has
that
\begin{align*}
  &\sup_{u,v}\sup_{\Gamma\in\{\mu,\lambda,\rho,\sigma,\alpha,\beta,\tau,\epsilon,\chi\}}
  \max\{\sum_{0}^1||\nablasl^i\Gamma||_{L^{\infty}(\mathcal{S}_{u,v})},
  \sum_{i=0}^2||\nablasl^i\Gamma||_{L^4(\mathcal{S}_{u,v})},
  \sum_{i=0}^3||\nablasl^i\Gamma||_{L^2(\mathcal{S}_{u,v})}\} \\
  &\hspace{1cm}+\sup_{u,v}\sup_{j=1,...,4}
  \{\sum_{i=0}^1||\nablasl^i\Sigma_j||_{L^{\infty}(\mathcal{S}_{u,v})},
  \sum_{i=0}^2||\nablasl^i\Sigma_j||_{L^4(\mathcal{S}_{u,v})},
  \sum_{i=0}^3||\nablasl^i\Sigma_j||_{L^2(\mathcal{S}_{u,v})}\}
  +\Delta_{\Phi}+\Delta_{\phi}\\
  &\hspace{2cm}\leq C(I,\Delta_{e_{\star},\Xi_{\star}},
  \Delta_{\Gamma_{\star}},\Delta_{\Sigma_{\star}},
  \Delta_{\Phi_{\star}},\Delta_{\phi_{\star}}),
\end{align*}
with $I\equiv[0,v_\bullet]$.
\end{theorem}

The proof of the above result is based on a lengthy bootstrap
argument. All the main ingredients for it have already been developed
in Papers~I and~II. The main task in this article is to verify that
the arguments follow through in the slightly different setting of the
problem of the local extension of null infinity. The various steps in
the proof are as follows:

\begin{itemize}
\item[(0)] Construct~$L^\infty$ estimates for the components of the
  frame and the conformal factor and its derivatives on the
  spheres~$\mathcal{S}_{u,v}$ in terms of initial data and the
  length~$\varepsilon$ of the short direction of integration.  These
  bounds, in turn, allow to control in a systematic manner the
  solutions of the transport equations implied by the CEFE along null
  directions.

\item[(i)] Construction of~$L^\infty$, $L^2$ and~$L^4$ estimates for
  the connection coefficients over the
  spheres~$\mathcal{S}_{u,v}$. These estimates require the assumption
  that the components of the curvature are bounded.

\item[(ii)] Show that the components of the curvature are bounded in
  the~$L^2$ norm on the spheres~$\mathcal{S}_{u,v}$. These bounds are
  given in terms of the initial conditions an the value of the
  curvature of the light cones~$\mathcal{N}_u$ and~$\mathcal{N}'_{v}$.

\item[(iii)] Show that the norms of the curvature on the light cones
  can be bounded in terms of the initial data.

\item[(iv)] Last slice argument. Make use of the estimates obtained in
  the previous steps to show that the solution to the evolution
  equations exists close to~$\mathcal{N}_\star$ as long as one has
  control of the data on this initial hypersurface.

\end{itemize}

\subsection{Estimates for the connection coefficients and the
  derivative of conformal factor }

In this section we provide a discussion of the first step of our
bootstrap argument and provide estimates for connection coefficients
and the derivative of conformal factor. In order to prove these
estimates it is assumed that the norms of the components of the
curvature spinors are bounded. It follows then that the short
range~$\varepsilon$ can be chosen such that connection coefficients
and the derivatives of the conformal factor can be controlled by the
norm of the initial data and the
norm~$\Delta_{\Phi}(\mathcal{S})$. The main tool in this estimation
are the transport equations satisfied by the various fields. Most of
the connection coefficients satisfy transport equations in both
the~$D$ and~$\Delta$ directions. Only for the connection
coefficients~$\tau$ and~$\chi$, we only have their long direction~$D$
equations. Crucially, however, these equations do not contain
quadratic terms and can basically be regarded as linear equations.

\smallskip 
The first step in the argumentation is to control the supremum norm of
the connection coefficients and the derivatives of the conformal
factor
---cf. Proposition~\ref{PaperII-Proposition:CEFEFirstEstimateConnectionSigma}
in Paper~II. The assumptions in this estimate are that there exists a
positive constant~$\Delta_{\Gamma,\Sigma}$
in~$\mathcal{D}^t_{u,v_\bullet}$ such that
\begin{align*}
  \sup_{u,v}||\{\mu, \lambda, \alpha, \beta, \epsilon, \rho, \sigma,
  \tau, \chi,\Sigma_1,\Sigma_2,\Sigma_3,\Sigma_4\}||_{L^\infty(\mathcal{S}_{u,v})}
  \leq \Delta_{\Gamma,\Sigma}\,,
\end{align*}
in a causal diamond and that, moreover,
\begin{align*}
  \sup_{u,v}||\nablasl^3\tau||_{L^2(\mathcal{S}_{u,v})}
  <\infty, \quad \Delta_{\Phi}(\mathcal{S})
  <\infty,\quad \Delta_{\Phi}<\infty,
  \quad \Delta_{\phi}(\mathcal{S})<\infty,\quad \Delta_{\phi}<\infty.
\end{align*}
Next, one constructs~$L^4$-estimates of the connection coefficients
and the derivative of conformal factor
---cf. Proposition~\ref{PaperII-Proposition:CEFESecondEstimateConnection}
in Paper~II. These estimates are needed to make use of the
Gagliardo-Nirenberg inequality in dealing with the non-linearities of
the evolution equations when constructing~$L^4$-estimates. This step
requires the further assumption that
\begin{align*}
\sup_{u,v}||\nablasl\{\mu, \lambda, \alpha, \beta, \epsilon, \rho,
\sigma,\Sigma_1,\Sigma_2,\Sigma_3,\Sigma_4\}
||_{L^4(\mathcal{S}_{u,v})}\leq\Delta_{\Gamma,\Sigma}.
\end{align*}
The last step in this process is a~$L^2$-estimate for the connection
coefficients and the derivative of conformal factor
---cf. Proposition~\ref{PaperII-Proposition:CEFEThirdEstimateConnection}
in Paper~II--- which is obtained without the need of any further
assumptions.

In order to estimate the components of the curvature, we
need~$L^2$-estimates of the connection coefficients and derivatives of
the conformal factor up to third order. This can be achieved by a
method similar to the one used to estimate the undifferentiated fields
---cf. Proposition~\ref{PaperII-Proposition:CEFEImprovedEstimates} in
Paper~II. The analysis described in the previous paragraphs can be
summarised as follows:

\begin{proposition}[\textbf{\em estimates for the~$L^\infty$,
        $L^4$ and~$L^2$ norms of the connection coefficients and the
      derivatives of the conformal factor to second derivative}]
\label{PropositionSummary:EstimatesConnection}
Assume 
\begin{align*}
\Delta_{\Phi}<\infty, \qquad  \Delta_{\phi}<\infty,
\end{align*}
in the truncated diamond~$\mathcal{D}_{u,v_\bullet}^{\,t}$. Then there
exists
\begin{align*}
\varepsilon_{\star}=\varepsilon_{\star}
(I,\Delta_{e_{\star},\Xi_{\star}},\Delta_{\Gamma_{\star}},
\Delta_{\Sigma_{\star}},\Delta_{\Phi},\Delta_{\phi},
\sup_{u,v}||\nablasl^3\tau||_{L^2(\mathcal{S}_{u,v})})
\end{align*}
such that for~$\varepsilon\leq\varepsilon_\star$, we have
\begin{align*}
  &\sup_{u,v}\sup_{\Gamma\in\{\mu,\lambda,\alpha,\beta,\epsilon,\rho,\sigma,\tau,\chi\}}
  \left(||\Gamma||_{L^{\infty}(\mathcal{S}_{u,v})}+\sum_{i=0}^1||\nablasl^i\Gamma
  ||_{L^4(\mathcal{S}_{u,v})}+\sum_{i=0}^2||\nablasl^i
  \Gamma||_{L^2(\mathcal{S}_{u,v})}\right)\\ 
  &\leq C(I,\Delta_{e_{\star},\Xi_{\star}},\Delta_{\Gamma_\star},
  \Delta_{\Sigma_\star},\Delta_{\Phi}(\mathcal{S}),\Delta_{\phi}(\mathcal{S})), \\
  &\sup_{u,v}\sup_{j=1,...,4}\left(||\Sigma_j||_{L^{\infty}(\mathcal{S}_{u,v})}
  +\sum_{i=0}^1||\nablasl^i\Sigma_j||_{L^4(\mathcal{S}_{u,v})}+\sum_{i=0}^2
  ||\nablasl^i\Sigma_j||_{L^2(\mathcal{S}_{u,v})}\right)
  \leq C(\Delta_{e_{\star},\Xi_{\star}},\Delta_{\Sigma_\star}), 
\end{align*}
in the truncated diamond~$\mathcal{D}_{u,v_\bullet}^{\,t}$.
\end{proposition}

Armed with~$L^2$-estimates for the connection coefficients and
derivatives of the conformal factor up to the second order, it is now
possible to show that the norms~$\Delta_{\Phi}(\mathcal{S})$
and~$\Delta_{\phi}(\mathcal{S})$ are finite ---see
Proposition~\ref{PaperII-Proposition:CEFEFirstEstimateRicciCurvature}
in Paper~II. More precisely, one has that:
\begin{proposition}[\textbf{\em boundedness of the components of
        the curvature}] Assume that
\begin{align*}
  \Delta_{\Phi}<\infty, \qquad \Delta_{\phi}<\infty, \qquad \sup_{u,v}||
  \nablasl^3\tau||_{L^2(\mathcal{S}_{u,v})}<\infty
\end{align*}
in the truncated diamond~$\mathcal{D}_{u,v_\bullet}^t$. Then there
exists
\begin{align*}
\varepsilon_{\star}=\varepsilon_{\star}(I,\Delta_{e_{\star},\Xi_{\star}},
\Delta_{\Gamma_{\star}},\Delta_{\Sigma_{\star}},\Delta_{\Phi_{\star}},
\Delta_{\phi_{\star}},\Delta_{\Phi},\Delta_{\phi},
\sup_{u,v}||\nablasl^3\tau||_{L^2(\mathcal{S}_{u,v})})
\end{align*}
such that for~$\varepsilon\leq\varepsilon_{\star}$, we have
\begin{align*}
  \Delta_{\Phi}(\mathcal{S})<\infty,
  \qquad \Delta_{\phi}(\mathcal{S})<\infty.
\end{align*}
\end{proposition}

With the results above, we gather all the estimates for
  connection coefficients and derivative of conformal factor:
\begin{proposition}[\textbf{\em estimates for the~$L^\infty$, $L^4$
    and~$L^2$ norms of the connection coefficients and the derivatives
    of the metric}]
\label{PropositionSummary:EstimatesConnectionDerivativeofConfFactor}
Assume 
\begin{align*}
\Delta_{\Phi}<\infty, \qquad  \Delta_{\phi}<\infty,
\end{align*}
in the truncated diamond~$\mathcal{D}_{u,v_\bullet}^{\,t}$. Then there
exists
\begin{align*}
\varepsilon_\star=\varepsilon_\star(I,\Delta_{e_{\star},\Xi_{\star}},\Delta_{\Gamma_{\star}},
\Delta_{\Sigma_{\star}},\Delta_{\Phi_{\star}},
\Delta_{\phi_{\star}},\Delta_{\Phi},\Delta_{\phi})
\end{align*}
such that for~$\varepsilon\leq\varepsilon_\star$, we have
\begin{align*}
  &\sup_{u,v}\sup_{\Gamma\in\{\mu,\lambda,\alpha,\beta,\epsilon\}}
  \left(\sum_{i=0}^1||\nablasl^i\Gamma||_{L^{\infty}
    (\mathcal{S}_{u,v})}+\sum_{i=0}^2||\nablasl^i\Gamma||_{L^4(\mathcal{S}_{u,v})}
  +\sum_{i=0}^3||\nablasl^i\Gamma||_{L^2(\mathcal{S}_{u,v})}\right)\\
  &\qquad\leq
  C(\Delta_{e_{\star},\Xi_{\star}},\Delta_{\Gamma_\star}), \\
  &\sup_{u,v}\left(||\{\rho,\sigma\}||_{L^{\infty}(\mathcal{S}_{u,v})}
  +\sum_{i=0}^1||\nablasl^i\{\rho,\sigma\}||_{L^4(\mathcal{S}_{u,v})}
  +\sum_{i=0}^2||\nablasl^i\{\rho,\sigma\}||_{L^2(\mathcal{S}_{u,v})} \right)
  \leq C(\Delta_{e_{\star},\Xi_{\star}},\Delta_{\Gamma_\star}), \\
  &\sup_{u,v}\left(||\{\tau,\chi\}||_{L^{\infty}(\mathcal{S}_{u,v})}
  +\sum_{i=0}^1||\nablasl^i\{\tau,\chi\}||_{L^4(\mathcal{S}_{u,v})}
  +\sum_{i=0}^2||\nablasl^i\{\tau,\chi\}||_{L^2(\mathcal{S}_{u,v})} \right)\\
  &\qquad\qquad
  \leq C(I,\Delta_{e_{\star},\Xi_{\star}},\Delta_{\Gamma_\star},\Delta_{\Sigma_\star},
  \Delta_{\Phi_\star},\Delta_{\phi_\star}),\\
  &\sup_{u,v}\left(||\nablasl\{\rho,\sigma\}||_{L^{\infty}(\mathcal{S}_{u,v})}
  +||\nablasl^2\{\rho,\sigma\}||_{L^4(\mathcal{S}_{u,v})}
  +||\nablasl^3\{\rho,\sigma\}||_{L^2(\mathcal{S}_{u,v})} \right)\\
  &\qquad\qquad
  \leq C(I,\Delta_{e_{\star},\Xi_{\star}},\Delta_{\Gamma_\star},
  \Delta_{\Sigma_\star},\Delta_{\Phi},\Delta_{\phi}),  \\
  &\sup_{u,v}\left(||\nablasl\{\tau,\chi\}||_{L^{\infty}(\mathcal{S}_{u,v})}
  +||\nablasl^2\{\tau,\chi\}||_{L^4(\mathcal{S}_{u,v})}
  +||\nablasl^3\{\tau,\chi\}||_{L^2(\mathcal{S}_{u,v})} \right)\\
  &\qquad\qquad
  \leq C(I,\Delta_{e_{\star},\Xi_{\star}},\Delta_{\Gamma_\star},
  \Delta_{\Sigma_\star},\Delta_{\Phi},\Delta_{\phi}), \\
  &\sup_{u,v}\sup_{j=1,...,4}\left(\sum_{i=0}^1
  ||\nablasl^i\Sigma_j||_{L^{\infty}(\mathcal{S}_{u,v})}
  +\sum_{i=0}^2||\nablasl^i\Sigma_j||_{L^4(\mathcal{S}_{u,v})}
  +\sum_{i=0}^3||\nablasl^i\Sigma_j||_{L^2(\mathcal{S}_{u,v})}\right)
  \leq C(\Delta_{e_{\star},\Xi_{\star}},\Delta_{\Sigma_\star}), 
\end{align*}
in the truncated diamond~$\mathcal{D}_{u,v_\bullet}^{\,t}$.
\end{proposition}

\subsection{The energy estimates for the curvature}

The next step in the bootstrap argument leading to the optimal local
existence result is to make use of the estimates provided by
Proposition~\ref{PropositionSummary:EstimatesConnection} to obtain
sharper energy estimates for the components of the Ricci and rescaled
Weyl curvature spinors. The hierarchical structure of the CEFE allows
to proceed with this estimation in a two-step process: first one looks
at the components of the Weyl tensor
---cf. Propositions~\ref{PaperII-Proposition:SecondMainEstimaterescaledWeylCurvature}
and~\ref{PaperII-Proposition:EstimatesDerivativesrescaledWeyl34} of
Paper~II. In the second step one estimates the components of the Ricci
tensor
---cf. Propositions~\ref{PaperII-Proposition:EstimatesDerivativesRiccigood}
and~\ref{PaperII-Proposition:EstimatesDerivativesRicci1222}. For both
the rescaled Weyl tensor and the Ricci tensor the analysis of most of
the components is straightforward. Only certain \emph{bad} components
require extra consideration ---the components~$\phi_3$ and~$\phi_4$ of
the Weyl tensor and the components~$\Phi_{12}$ and~$\Phi_{22}$ of the
Ricci tensor. The final result of this analysis is the following
proposition estimating the components of the curvature in terms of the
initial data. The key ingredient in this proposition is the assumption
that the curvature is bounded.

\begin{proposition} [\textbf{\em control of the components of the
        curvature in terms of the initial data}]
\label{Proposition:FinalEstimateRicciCurvature}
Suppose we are given a solution to the vacuum CEFE's in Stewart's
gauge arising from data for the CIVP satisfying
\begin{align*}
 \Delta_{e_{\star},\Xi_{\star}},\;\Delta_{\Gamma_\star}, \;\Delta_{\Sigma_\star},
  \;\Delta_{\Phi_\star}\;\Delta_{\phi_\star} <\infty,
\end{align*}
with the solution itself satisfying
\begin{align*}
 & \sup_{u,v}||\{\mu, \lambda, \alpha, \beta, \epsilon, \rho, \sigma,
  \tau, \chi,\Sigma_i\}||_{L^\infty(\mathcal{S}_{u,v})}< \infty\,,\quad
  \sup_{u,v}||\nablasl\{\mu, \lambda, \alpha, \beta, \epsilon, \rho,
  \sigma,\Sigma_i\}
  ||_{L^4(\mathcal{S}_{u,v})}<\infty\,,\\
&  \sup_{u,v}||\nablasl^2\{\mu,
  \lambda, \alpha, \beta, \epsilon, \rho, \sigma,
  \tau,\Sigma_i\}||_{L^2(\mathcal{S}_{u,v})}<\infty\,,\quad
  \sup_{u,v}||\nablasl^3\{\mu,\lambda,\alpha,\beta,\epsilon,\tau,\Sigma_i\}
  ||_{L^2(\mathcal{S}_{u,v})}<\infty\,,\\
 & \qquad\qquad\qquad\qquad\Delta_{\Phi}(\mathcal{S})<\infty\,, \quad
  \Delta_{\Phi}<\infty\,,\quad \Delta_{\phi}(\mathcal{S})<\infty\,, \quad
  \Delta_{\phi}<\infty,
\end{align*}
on some truncated causal
diamond~$\mathcal{D}_{u,v_\bullet}^{\,t}$. Then there
exists~$\varepsilon_\star=
\varepsilon_\star(I,\Delta_{e_{\star},\Xi_{\star}},\Delta_{\Gamma_\star},
\Delta_{\Sigma_\star},\Delta_{\Phi_\star},\Delta_{\phi_{\star}})$ such
that for~$\varepsilon_\star\leq\varepsilon$ we have
\begin{align*}
  &\Delta_{\Phi}<C_1(I,\Delta_{e_{\star},\Xi_{\star}},\Delta_{\Gamma_{\star}},
  \Delta_{\Sigma_{\star}},\Delta_{\Phi_{\star}},\Delta_{\phi_{\star}}),
  \\
  &\Delta_{\phi}\leq
  C_2(I,\Delta_{e_{\star},\Xi_{\star}},\Delta_{\Gamma_{\star}},
  \Delta_{\Sigma_{\star}},\Delta_{\Phi_{\star}},
  \Delta_{\phi_{\star}}).
\end{align*}
\end{proposition}

\subsection{Last slice argument}

The estimates discussed in the previous subsections can be used, in
turn, to show that the solution to the conformal Einstein field
equations exist on a rectangular domain of the form
\begin{align*}
\mathcal{D} \equiv \{ 0\leq u \leq \varepsilon, \; 0\leq v \leq v_\bullet\}, 
\end{align*}
with~$v_\bullet$ such that~$\Xi|_{v=v_{\bullet}}=0$. Accordingly, the
set~$\{ v=v_\bullet \}\cap \mathcal{D}$ corresponds to a portion of
future null infinity~$\mathscr{I}^+$. The strategy to show this result
is similar to the one used in Papers~I and~II and is based on a
\emph{last slice argument}. In this scheme one argues by contradiction
and assumes that the solution does not fill the whole
domain~$\mathcal{D}$. Accordingly, there must exist a hypersurface
(the last slice) which bounds the domain of existence of the
solution. The estimates constructed in the previous subsections allow
then to show that in this last slice the solution and its derivatives
are bounded so that it is possible to formulate a (standard) initial
value problem for the conformal Einstein field equations to show that
the solution extends beyond the last slice ---thus resulting in a
contradiction.

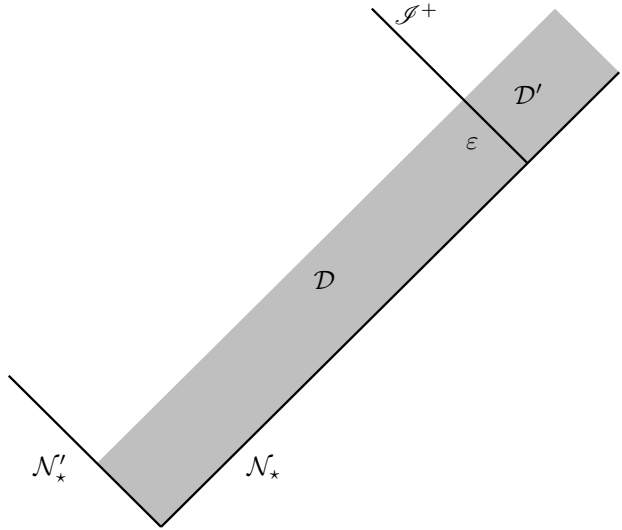
\begin{figure}[t]
\centering
\includegraphics[width=0.6\textwidth]{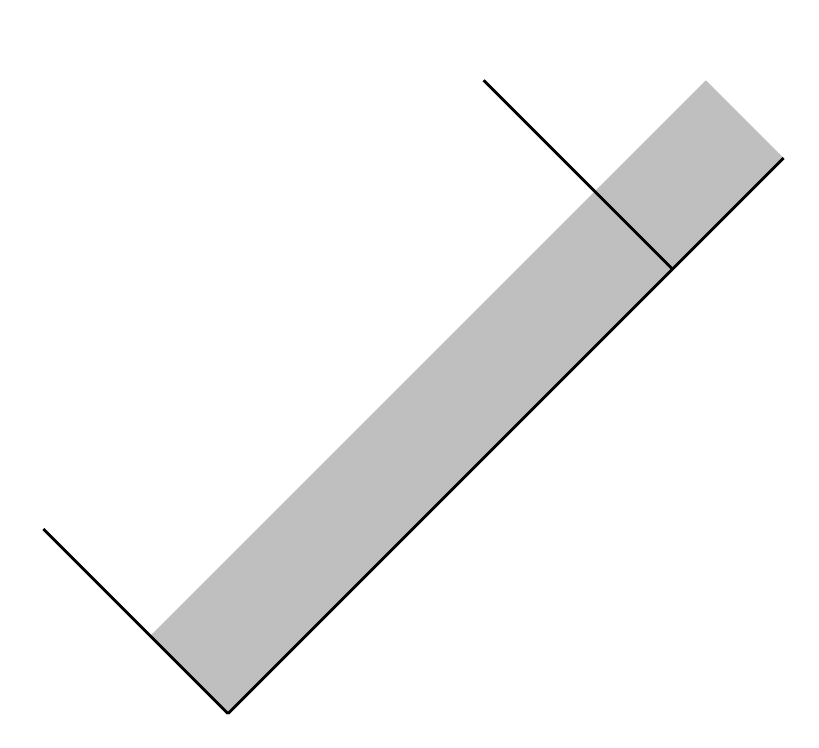}
\put(-95,200){$\mathscr{I}^+$}
\put(-68,153){$\varepsilon$}
\put(-125,100){$\mathcal{D}$}
\put(-50,170){$\mathcal{D}'$}
\put(-230,30){$\mathcal{N}'_\star$}
\put(-150,30){$\mathcal{N}_\star$}
\caption{Extension of the initial data on~$\mathcal{N}_\star$. By
  causality the choice of the extension of the data
  beyond~$\mathscr{I}^+$ does not influence the solution of the causal
  domain~$\mathcal{D}$.\label{Figure:ExtendedDataLastSlice}}
\end{figure}

As the workings of the last slice argument have been discussed in
detail in Paper~I ---see section~\ref{PaperI-Section:LastSlice} of
this reference--- here we focus on the necessary modifications. As the
main purpose of the present analysis is to ensure that one recovers a
portion of future null infinity, in order to ensure existence of the
solution to the CEFE on the domain~$\mathcal{D}$ one actually needs to
show existence in a slightly larger domain. This is because the
existence domains are given in terms of open sets. As the CEFE are
regular at the sets where~$\Xi=0$, one can consider an initial
hypersurface~$\mathcal{N}_\star$ which extends
beyond~$\mathscr{I}^+$. The basic initial data on~$\mathcal{N}_\star$
as described in Proposition~\ref{Lemma3} can be extended in an
arbitrary but controlled manner beyond the intersection of null
infinity with~$\mathscr{I}^+$ up to, say~$v_\bullet+\frac{1}{10}$, in
such a way that it coincides with the original data
for~$v\in[0,v_\bullet]$ ---see
Figure~\ref{Figure:ExtendedDataLastSlice}. In particular we require
that the extension is such that the norms~$\Delta_{e_\star,
  \Xi_\star}$, $\Delta_{\Gamma_\star}$, $\Delta_{\Sigma_\star}$,
$\Delta_{\Phi_\star}$ and~$\Delta_{\phi_\star}$ which have a
contribution along~$\mathcal{N}_\star$ are finite. Using this extended
data on~$\mathcal{N}_\star$ together with the data
on~$\mathcal{N}'_\star$ and $\mathcal{S}_\star$ one can compute the
full initial data set for the conformal evolution equations. The last
slice argument as discussed in Papers~I and~II can then be used to
ensure existence on
\begin{align*}
\mathcal{D}'\equiv \{ 0\leq u \leq \varepsilon, \; 0
\leq v \leq v_\bullet +\tfrac{1}{10}\} \supset \mathcal{D}.
\end{align*}
As a consequence of Lemma~\ref{Lemma:StructureConformalBoundary}, one
has that the set defined by the condition~$v=v_\bullet$ is a null
hypersurface and, accordingly, our domain of existence contains a
portion of~$\mathscr{I}^+$. Finally, observe that by causality the
solution on~$\mathcal{D}$ is independent of the choice of extended
data on~$\{ u=0, \, v\in(v_\bullet,v_\bullet +\tfrac{1}{10} \}$
---that is, $\mathscr{I}^+$ is the Cauchy horizon of the data on~$\{v=0, \;
u\in[0,\varepsilon\}\cup \{ u=0, \; v\in[0,v_\bullet] \}$.
  
\section{Application: stability of the Minkowski spacetime from a
  Cauchy-characteristic initial value problem}

In this section we discuss an application of the local extendibility
problem of null infinity to the stability of the Minkowski spacetime
in a Cauchy-characteristic setting. This argument relies crucially on
Friedrich's semi-global existence and stability result of the Minkowski
spacetime from hyperboloidal data ---see~\cite{Fri86a}; see
also~\cite{LueVal09}. The strategy in the proof is to use the local
extendibility result of null infinity proven earlier in this article
together with the local existence result for the standard Cauchy
problem for the conformal Einstein field equations to obtain a
development of the Cauchy-characteristic initial data on which one can
pass a hyperboloidal hypersurface. If the Cauchy-characteristic
initial data is suitably and sufficiently close to data for the
Minkowski spacetime then one will be in a situation in which the
semi-global stability of the Minkowski spacetime can be used.

\begin{figure}[t]
\centering
\includegraphics[width=0.7\textwidth]{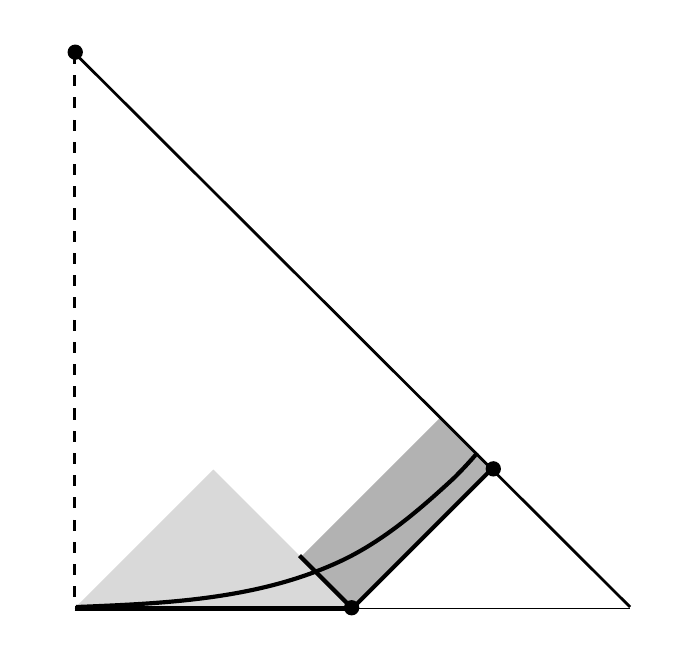}
\put(-128,140){$\mathscr{I}^+$}
\put(-115,82){$\mathcal{H}_\star$}
\put(-113,50){$\mathcal{N}_\star$}
\put(-200,15){$\mathcal{K}_\star$}
\put(-152,37){$\mathcal{N}_\star'$}
\put(-220,60){$D^+(\mathcal{K}_\star)$}
\put(-220,140){$D^+(\mathcal{H}_\star)$}
\put(-185,70){$D^+(\mathcal{N}_\star'\cup\mathcal{N}_\star)$}
\caption{Schematic depiction of the set-up for the proof of the
  stability of the Minkowski spacetime in a Cauchy-characteristic
  setting. Cauchy data for the conformal Einstein field equations is
  provided in the compact spacelike domain~$\mathcal{K}_\star$ while
  characteristic data consistent with Lemma~\ref{Lemma3} is provided
  on the null hypersurface~$\mathcal{N}_\star$. This null hypersurface
  intersects null infinity,~$\mathscr{I}^+$. The
  development~$D^+(\mathcal{K}_\star)$ (light-grey region) gives rise
  to characteristic data on the Cauchy horizon~$\mathcal{N}'_\star =
  H(\mathcal{K}_\star)$. The theory of the local extendibility of null
  infinity developed earlier in the article ensures the existence of a
  solution to the conformal equations in the causal
  diamond~$D(\mathcal{N}_\star\cup\mathcal{N}_\star')$. Crucially, the
  boundary of this causal diamond includes a portion
  of~$\mathscr{I}^+$. On~$D^+(\mathcal{K}_\star)\cup
  D(\mathcal{N}_\star\cup\mathcal{N}_\star')$ one considers a
  hyperboloidal hypersurface~$\mathcal{H}_\star$. If the
  Cauchy-characteristic initial data is assumed to be suitably close
  to data for the the Minkowski spacetime, then the hyperboloidal data
  induced on~$\mathcal{H}_\star$ will also be suitably close to
  Minkowski hyperboloidal data. Friedrich's semi-global existence
  result then ensures that~$D(\mathcal{H}_\star)$ is geodesically
  complete and has the same global structure as the Minkowski
  spacetime ---in particular, the generators of~$\mathscr{I}^+$
  intersect a point~$i^+$, spatial
  infinity. \label{Fig:StabilityMinkowski}}
\end{figure}
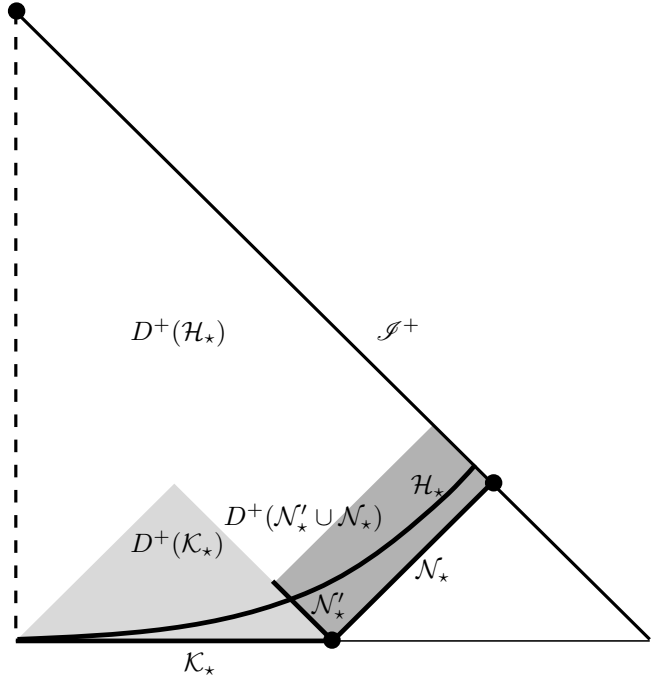

\subsection{Set-up}
\label{MinkowskiSetUp}

In the Cauchy-characteristic initial value problem it is assumed one
is provided with standard Cauchy initial data for the Einstein field
equations in a compact domain~$\mathcal{K}_\star$ of a spacelike
hypersurface~$\mathcal{S}_\star$. In the following it is assumed
that~$\mathcal{K}_\star= \mathcal{B}_{r_\star}$
with~$\mathcal{B}_{r_\star}$ a solid $3$-dimensional ball of
radius~$r_\star$ as measured by the metric $\bmh$ of
$\mathcal{S}_\star$ ---in particular,~$\p \mathcal{B}_{r_\star}
\approx \mathbb{S}^2$. Moreover, it is assumed that the hypersurface
$\mathcal{S}_\star$ is intersected at~$\p \mathcal{B}_{r_\star}$
by a null hypersurface~$\mathcal{N}_\star$ intersecting future null
infinity on a cut~$\mathcal{C}_\star \approx \mathbb{S}^2$. A sketch
of the the set-up is given in Figure~\ref{Fig:StabilityMinkowski}.

\smallskip
The initial data on~$\mathcal{K}_\star$ is assumed to satisfy the
Einstein constraint equations. These imply, in turn, a solution to the
constraints implied by the conformal Einstein field equations
on~$\mathcal{K}_\star$. In the following, it is assumed that the
initial data for the conformal Einstein equations
on~$\mathcal{K}_\star$, to be denoted by~$\mathbf{u}_\star$, differs
from standard (time symmetric) Cauchy data for the Minkowski
spacetime, to be denoted by~$\mathring{\mathbf{u}}_\star$ by at
most~$\delta>0$ in the standard Sobolev norm of order~$s\geq 4$
with $s$ sufficiently large. That is,
\begin{align*}
\parallel \mathbf{u}_\star
-\mathring{\mathbf{u}}_\star\parallel_{s,\mathcal{K}_\star} \leq \delta
\end{align*}
It follows from the theory of symmetric hyperbolic systems that the
development of this data will be of class~$C^{s-2}$ over its future
domain of dependence~$D^+(\mathcal{K}_\star)$ ---see
e.g.~\cite{Fri86b,CFEBook}. In particular if the initial data has
finite~$\parallel \;\;\parallel_{s,\mathcal{K}_\star}$-norm for
every~$s\geq 4$, then the solution is of class~$C^\infty$
over~$D^+(\mathcal{K}_\star)$ ---this will be assumed in the following
for simplicity of presentation.  On~$\p\mathcal{B}_{r_\star}$ it is
assumed that~$\mathbf{u}_\star$ matches smoothly with characteristic
data initial data on~$\mathcal{N}_\star$ ---in the following this data
will be denoted by~$\mathbf{v}_\star$. Similarly, one also considers
characteristic data $\mathring{\mathbf{v}}_\star$ for the Minkowski
spacetime satisfying the assumptions and gauge conditions of
Lemma~\ref{Lemma3} ---an explicit construction of this data of the
Minkowski spacetime is given in
Appendix~\ref{Appendix:ConformalMinkowskiStewartGauge}.

Now, the improved existence result for the characteristic problem
contained in Theorem~\ref{Theorem:LocalExtensionNullInfinity} does not
directly provide a statement of Cauchy stability which could be used
in the analysis of the non-linear stability of the Minkowski
spacetime. The reason for this is the presence of the ``number $1$''
in the definitions of the quantities $\Delta_{e_\star,\Xi_\star}$,
$\Delta_{\Gamma_\star}$, $\Delta_{\Phi_\star}$
and~$\Delta_{\phi_\star}$ introduced in
Subsection~\ref{Subsection:NormsAndMainResult}. This ``$1$'' was
originally introduced in the definition of the analogue quantities
quantities in \cite{Luk12} as a safety valve to ensure that the
arguments run through even in the case of trivial data. This feature
has the consequence that the quantities~$\Delta_{e_\star,\Xi_\star}$,
$\Delta_{\Gamma_\star}$, $\Delta_{\Phi_\star}$
and~$\Delta_{\phi_\star}$ are, strictly speaking, not norms. As it
will be shown below, the number ``$1$'' can be removed from the
definition if one has more information about the type of solution one
wants to construct ---as, for example, closeness to the Minkowski
spacetime.

\subsection{Cauchy stability for the characteristic initial value problem}

In this section it is discussed how to show that if the solution to
the CIVP on the causal diamond of
Theorem~\ref{Theorem:LocalExtensionNullInfinity} arises from
characteristic data~$\mathring{\mathbf{v}}_\star$ for the Minkowski
spacetime, then the solution on the existence diamond is also suitably
close to Minkowski data. The precise notion of \emph{closeness}
follows naturally from the strategy of the proof followed in this
article for our main theorem.

Start by recalling that as a consequence of
Theorem~\ref{Theorem:LocalExtensionNullInfinity} we have already
proved the existence of solutions to the CEFE for given arbitrary data
on~$\mathcal{N}_\star$ and~$\mathcal{N}'_\star$. As in the previous
subsection assume that on the existence diamond one has two
solutions~$\mathbf{v}$ and~ $\mathring{\mathbf{v}}$ ---the later
corresponding to the Minkowski spacetime as give, say, in
Appendix~\ref{Appendix:ConformalMinkowskiStewartGauge}. In the
following we use the quantities~$\xi_{\bme}$, $\xi_{\Xi}$,
$\xi_{\Gamma}$, $\xi_{\Sigma}$, $\xi_{\phi}$ and~$\xi_{\Phi}$ encode
the difference of quantities field unknowns between the
\emph{perturbed} and the Minkowski spacetime. For
example~$\xi_{\mu}\equiv\mu-\mathring{\mu}$ where $\mathring\mu$ gives
the value of the NP spin connection coefficient on the Minkowski
spacetime. For these \emph{difference quantities} one can define true
norms~$\Delta_{\xi_{\bme_{\star},\Xi_{\star}}}$,
$\Delta_{\xi_{\Gamma_{\star}}}$, $\Delta_{\xi_{\Sigma_{\star}}}$,
$\Delta_{\xi_{\Phi_{\star}}}$ and $\Delta_{\xi_{\phi_{\star}}}$ and
assume such norms are controlled by a small constant $\delta$. These
norms are defined in analogous manner to the quantities in
Section~\ref{Subsection:NormsAndMainResult} but without the
``number~$1$'' and with the understanding that the derivatives
appearing in them are operators on the perturbed spacetime. For
example, the norm for the initial value of the differences between
connection coefficients is given by
\begin{align*}
 \Delta_{\xi_{\Gamma_{\star}}}\equiv\sup_{\mathcal{S}_{u,v}\subset\mathcal{N}_\star,\mathcal{N}'_\star}
 \sup_{\Gamma\in\{\mu,\lambda,\rho,\sigma,\alpha,\beta,\tau,\epsilon\}}
 \max\{\sum_{i=0}^1||\nablasl^i\xi_{\Gamma}||_{L^{\infty}(\mathcal{S}_{u,v})},\sum_{i=0}^2
 ||\nablasl^i\xi_{\Gamma}||_{L^4(\mathcal{S}_{u,v})},
 \sum_{i=0}^3
 ||\nablasl^i\xi_{\Gamma}||_{L^2(\mathcal{S}_{u,v})}\}.
\end{align*}
Similarly, the norm for the initial value of the difference of the
components of the Ricci curvature given is given by
\begin{align*}
  &\Delta_{\xi_{\Phi_{\star}}}\equiv
    \sup_{\mathcal{S}_{u,v}\subset\mathcal{N}_\star,\mathcal{N}'_\star}
    \sup_{\Phi\in\{\Phi_{00},\Phi_{01},\Phi_{02},\Phi_{11},\Phi_{12}\}}
    \max\{\sum_{i=0}^1||\nablasl^i\xi_{\Phi}||_{L^4(\mathcal{S}_{u,v})},
    \sum_{i=0}^2||\nablasl^i\xi_{\Phi}||_{L^2(\mathcal{S}_{u,v})}\} \\
  &+\sum_{i=0}^3\sup_{\Phi\in\{\Phi_{00},\Phi_{01},\Phi_{02},\Phi_{11},\Phi_{12}\}}
    ||\nablasl^i\xi_{\Phi}||_{L^2(\mathcal{N}_\star)}
    +\sup_{\Phi\in\{\Phi_{01},\Phi_{02},\Phi_{11},\Phi_{12},\Phi_{22}\}}
    ||\nablasl^i\xi_{\Phi}||_{L^2(\mathcal{N}'_\star)}.
\end{align*}
Equations for the \emph{difference fields} can be readily computed by
subtraction of the relevant evolution equations. The structure of the
resulting equations resembles those of the CEFE with additional terms
corresponding to the products between the background (i.e. Minkowski)
and difference terms. For example, from the structure equation
\begin{align*}
\Delta\lambda=-2\mu\lambda-\Xi\phi_4,
\end{align*}
 one obtains that 
\begin{align*}
  \Delta\xi_{\lambda}\equiv\Delta(\lambda-\mathring{\lambda})=
  -2\xi_{\mu}\xi_{\lambda}-\xi_{\Xi}\xi_{\phi_4} -2\mathring{\mu}\xi_{\lambda}
  -2\mathring{\lambda}\xi_{\mu}-\mathring{\Xi}\xi_{\phi_4}-\mathring{\phi}_4
  \xi_{\Xi}-\frac{\mathring{\Delta}\mathring{\lambda}}{\mathring{Q}}\xi_Q,
\end{align*}
where~$\mathring{\phi}_4$ denotes the component of rescaled Weyl
spinors on the Minkowski spacetime (zero!) and~$\mathring{\Delta}$ is
the derivative along~$\bmn$ on Minkowski. More generally, writing the
structure equations in schematic form as
\begin{align*}
&D\Gamma-\delta\Gamma=\Gamma\Gamma+\Xi\phi+\Phi,\\
&\Delta\Gamma-\delta\Gamma=\Gamma\Gamma+\Xi\phi+\Phi, 
\end{align*}
it follows that the associated \emph{difference equations} are of the
form
\begin{align*}
D\xi_\Gamma-\delta\xi_\Gamma=&\xi_\Gamma\xi_\Gamma+\xi_\Xi\xi_\phi+\xi_\Phi \\
&+\mathring{\Gamma}\xi_\Gamma+\mathring{\Xi}\xi_\phi+\mathring{\phi}\xi_\Xi+\xi_{P^{\mathcal{A}}}
\partial_{\mathcal{A}}\mathring{\Gamma}-\xi_{C^{\mathcal{A}}}\partial_{\mathcal{A}}\mathring{\Gamma} ,  \\
\Delta\xi_\Gamma-\delta\xi_\Gamma=&\xi_\Gamma\xi_\Gamma+\xi_\Xi\xi_\phi+\xi_\Phi \\
&+\mathring{\Gamma}\xi_\Gamma+\mathring{\Xi}\xi_\phi+\mathring{\phi}\xi_\Xi+\xi_{P^{\mathcal{A}}}
\partial_{\mathcal{A}}\mathring{\Gamma}-\frac{\mathring{\Delta}\mathring{\lambda}}{\mathring{Q}}\xi_Q.
\end{align*}
Notice that is the background solution corresponds to the Minkowski
spacetime then the Weyl curvature terms $\mathring{\phi}$ actually
vanish. From the general structure of these equations it follows that
it is possible to obtain estimates for the difference quantities
associated to the spin connection coefficients making use of an
argument analogous to that used for the actual connection
coefficients. For~$\xi_\Sigma$ and most of the
differences~$\xi_\Gamma$ one can make use of their $\Delta$-equations
construct the required estimates. It follows that as long as the short
range~$\varepsilon$ is sufficiently small, their norms can be bounded
by 3 times the size of the initial data. 

The analysis of the third order derivatives of~$\xi_\tau$, $\xi_\chi$
and~$\xi_\rho$ $\xi_\sigma$ requires the use of the
associated~$D$-equations and, thus, require integration along the long
direction so as to obtain a Gr\"onwall-type inequality of the form
\begin{align*}
||f||_{L^p(\mathcal{S}_{u,v})}\leq C(I,\mathcal{O})
\left(||f||_{L^p(\mathcal{S}_{u,0})}+\int_0^v
||Df||_{L^p(\mathcal{S}_{u,v'})}\mathrm{d}v' \right),
\end{align*}
where~$\mathcal{O}$ denote a constant related to the spin connection
coefficient~$\rho$ on the perturbed spacetime. The first term on the
right hand side corresponds to initial data for the differences~$\xi$
and, thus, is bounded, by assumption, by~$\delta$. The second term can
be estimated using the~$D$-equation; each term in this equation can be
bounded by a constant related to the value of the field on the
Minkowski spacetime times~$\delta$. It follows then that the norm of
these particular difference quantities in the solution rectangle can
be bounded by~$C\delta$ where~$C$ is a constant related to the data
for the background and perturbed spacetimes. The evolution equations
for the difference fields in the short direction can be analysed in
analogous manner to what is done for the main
Theorem~\ref{Theorem:LocalExtensionNullInfinity} ---in this case the
argument requires a bootstrap assumption.

Finally, an analogous strategy can be adopted, \emph{mutatis mutandi},
for the differences fields associated to the components of the Ricci
and rescaled Weyl curvature tensors. It is found that in the existence
diamond these difference fields are bounded by $C\delta$. Proceeding
in this way it is possible to show that in the existence diamond the
\emph{norms} $\Delta_{\xi_{\bme,\Xi}}$, $\Delta_{\xi_{\Gamma}}$,
$\Delta_{\xi_{\Sigma}}$, $\Delta_{\xi_{\Phi}}$
and~$\Delta_{\xi_{\phi}}$ can be controlled by~$C\delta$ where $C$
denotes again a constant related to the initial data. \emph{This
  observation is a statement of Cauchy stability for the development of
  characteristic initial data which is close to data for the Minkowski
  spacetime}.

\subsection{Statement of the stability result}

Following the discussion from the previous subsection, it is assumed
that the characteristic initial data $\mathbf{v}_\star$ differs from
the Minkowski characteristic initial data $\mathring{\mathbf{v}}
_\star$ in the norms $\Delta_{\xi_{\bme_\star,\Xi_\star}}$,
$\Delta_{\xi_{\Gamma_\star}}$, $\Delta_{\xi_{\Sigma_\star}}$,
$\Delta_{\xi_{\Phi_\star}}$ and~$\Delta_{\xi_{\phi_\star}}$ by at most
$\delta>0$.  At the intersection of the initial Cauchy and
characteristic hypersurfaces, $\mathcal{K}_\star\cap\mathcal{N}_\star
= \p \mathcal{B}_{r_\star}$, it is assumed that the Cauchy data for
the conformal Einstein field equations and the characteristic data
on~$\mathcal{N}_\star$ match smoothly.

\smallskip
Given the above setting, one has the following stability result:

\begin{theorem}
Let~$\mathcal{K}_\star$ and~$\mathcal{N}_\star$ as in Subsection
\ref{MinkowskiSetUp}. Given smooth initial data for the Einstein field
equations on~$\mathcal{K}_\star$ matching smoothly
on~$\p\mathcal{K}_\star$ to smooth characteristic data
on~$\mathcal{N}_\star$ extending smoothly to null infinity. If the
data on~$\mathcal{K}_\star\cup\mathcal{N}_\star$ is suitably close to
Cauchy-characteristic initial data for the Minkowski solution (in the
natural norms for~$\mathcal{K}_\star$ and~$\mathcal{N}_\star$,
respectively) then the future
development~$D^+(\mathcal{K}_\star\cup\mathcal{N}_\star)$ is
geodesically complete and has the same global structure than the
Minkowski spacetime. In particular, it has a smooth conformal
extension with a conformal boundary~$\mathscr{I}^+$ with complete null
generators intersecting at a point~$i^+$ representing future timelike
infinity.
\end{theorem} 

\begin{remark}
{\em For simplicity of presentation, the above theorem has been stated
  in the smooth (i.e.~$C^\infty$) class. However, a detailed analysis
  of the proof, and, in particular, of the argument leading to the
  result on the local extension of future null infinity,
  Theorem~\ref{Theorem:LocalExtensionNullInfinity} should lead to
  sharp statements in terms of the assumed regularity of the initial
  conditions and the resulting regularity of the solution. This
  analysis, is not be pursued presently.}
\end{remark}

\subsection{Proof}

Given the set-up described in the previous paragraphs, the argument to
show the stability of the Minkowski spacetime from
Cauchy-characteristic data proceeds as follows:

\begin{itemize}
\item[(i)] The Cauchy data given on~$\mathcal{K}_\star$ gives rise to
  a future development~$D^+(\mathcal{K}_\star)$. As this data does not
  cover the whole of a Cauchy hypersurface, the development has a
  Cauchy horizon~$H^+(\mathcal{K}_\star)$. The general theory of
  Lorentzian theory ensures that~$H^+(\mathcal{K}_\star)$ is a smooth
  null hypersurface~\cite{HawEll73}. Moreover, choosing the existence
  time of the development of~$\mathcal{K}_\star$ sufficiently small
  one ensure that the solution to the conformal Einstein field
  equations on~$D^+(\mathcal{K}_\star)$ extends smoothly
  to~$H^+(\mathcal{K}_\star)$. Now,
  setting~$\mathcal{N}_\star'=H^+(\mathcal{K}_\star)$ the restriction
  of the solution to the conformal Einstein field equations implies
  characteristic data on~$\mathcal{N}'_\star$ in the sense of
  Lemma~\ref{Lemma3} which is suitably close to characteristic data
  for the Minkowski spacetimes. Observe that for exact Minkowski data
  one has~$\phi_4=0$ on~$\mathcal{N}'_\star$.

\item[(ii)] The full characteristic data on~$\mathcal{N}_\star \cup
  \mathcal{N}'_\star$ give rise to
  development~$D^+(\mathcal{N}'_\star\cup\mathcal{N}_\star)$ in the
  form of a causal diamond which, by the theory developed in the
  previous sections of this article has one side coinciding with a
  portion of future null infinity~$\mathscr{I}^+$. A detailed
  inspection of the last slice argument shows, in particular, that if
  the data on the initial characteristic
  hypersurface~$\mathcal{N}_\star \cup \mathcal{N}'_\star$ is smooth,
  then the development~$D^+(\mathcal{N}'_\star\cup\mathcal{N}_\star)$
  is also smooth.

\item[(iii)] Now, the \emph{Cauchy-characteristic development}
  $D^+(\mathcal{K}_\star)\cup
  D^+(\mathcal{N}'_\star\cup\mathcal{N}_\star)$ contains a smooth
  hyperboloid~$\mathcal{H}_\star$. An explicit example can be given as
  follows: consider the (physical) Minkowski
  spacetime~$(\mathbb{R}^4,\tilde{\bmeta})$ in spherical
  coordinates~$(\tilde{t},\tilde{r},x^\mathcal{A})$. Then the upper
  sheet of the hypersurface given by the condition
\begin{align}
\frac{(\tilde{t}+a)^2}{a^2} -\frac{\tilde{r}^2}{1-a^2}=1
\label{Hyperboloid}
\end{align}
with~$a=\tfrac{1}{2}\sqrt{5}-\tfrac{1}{2}$ is an hyperboloidal
hypersurface passing through the origin. In particular, for large $r$
the hyperboloid asymptotes the null hypersurface described by the
condition
\begin{align*}
\tilde{t} -\tilde{r} =0.
\end{align*}
The conformal factor
\begin{align*}
\Theta =\cos t + \cos r
\end{align*}
where 
\begin{align*}
\tilde{t} = \frac{\sin t}{\cos t + \cos r}, \qquad \tilde{r}
=\frac{\sin r}{\cos t + \cos r},
\end{align*}
with~$t\in[-\pi,\pi]$, $r\in [0,\pi]$ gives an embedding of the
Minkowski spacetime into the Einstein cylinder. Now, as we are working
with a perturbation of the Minkowski spacetime, the
coordinates~$(t,r,x^{\mathcal{A}})$ can also be used, in a slight
abuse of notation, as coordinates of
the perturbed spacetime. The
upper sheet of the hypersurface described by the
condition~\eqref{Hyperboloid} expressed in terms of the
coordinates~$(t,r)$ given above is also an hyperboloid
on~$D^+(\mathcal{K}_\star)\cup
D^+(\mathcal{N}'_\star\cup\mathcal{N}_\star)$.

\item[(iv)] As the initial data on~$\mathcal{K}_\star \cup
  \mathcal{N}_\star$ is assumed to be suitably close to
  Cauchy-characteristic initial data for the Minkowski spacetime, then
  the solution to the conformal Einstein field equations
  on~$D^+(\mathcal{K}_\star)\cup D^+(\mathcal{N}_\star)$ is also
  suitably close to the Minkowski solution. By construction, in the
  compound domain~$D^+(\mathcal{K}_\star)\cup
  D^+(\mathcal{N}'_\star\cup\mathcal{N}_\star)$, the resulting
  solution to the conformal Einstein field equations is smooth and
  \emph{controlled} by the Cauchy-characteristic initial data
  on~$\mathcal{K}_\star\cup \mathcal{N}_\star$. In particular, due to
  the smoothness of the solution, one gets control at the level of the
  supremum norm over the whole of any
  derivatives~$D^+(\mathcal{K}_\star)\cup
  D^+(\mathcal{N}'_\star\cup\mathcal{N}_\star)$. Observe that as the
  domain is compact then the supremum of any of the conformal fields
  coincides with the maximum.

\item[(v)] As the solution on~$D^+(\mathcal{K}_\star)\cup
  D^+(\mathcal{N}'_\star\cup\mathcal{N}_\star)$ is smooth, it follows
  that,~$\mathbf{w}_\star$, the implied hyperboloidal initial data for
  the conformal Einstein field equations on~$\mathcal{H}_\star$ is
  smooth. Moreover, as in the conformal picture the
  hyperboloid~$\mathcal{H}_\star$ is a compact 3-manifold, it follows
  then that all the derivatives of the induced hyperboloidal data are
  in~$L^2(\mathcal{H}_\star)$ so that~$\mathbf{w}_\star\in
  H^s(\mathcal{H}_\star)$ for all~$s\geq 4$. By
  choosing~$\varepsilon>0$ in the Cauchy-characteristic initial data
  sufficiently small, one can make sure that~$\mathbf{w}_\star$ is
  suitably close to hyperboloidal data for the Minkowski solution.

\item[(vi)] Applying Friedrich's semi-global stability result for the
  Minkowski spacetime to the hyperboloidal data~$\mathbf{w}_\star$
  on~$\mathcal{H}_\star$ it follows that one obtains a smooth future
  development~$D^+(\mathcal{H}_\star)$ with Cauchy
  horizon~$H^+(\mathcal{H}_\star)$ which an be identified with future
  null infinity $\mathscr{I}^+$. The null hypersurface~$\mathscr{I}^+$
  has generators which are future complete and which intersect at a
  point~$i^+$ ---timelike infinity.

\item[(vii)] The solution to the conformal Einstein field equations
  on~$D^+(\mathcal{K}_\star)$ implies, whenever~$\Xi\neq 0$ a solution
  to the Einstein field equation which is future geodesically complete
  and with the same global asymptotic structure than the Minkowski
  spacetime.

\end{itemize}

\section*{Acknowledgements}

DH was supported by the FCT (Portugal) IF Program IF/00577/2015, by
project PTDC/MAT-APL/30043/2017 and~Project~No.~UIDB/00099/2020. PZ
acknowledges the support of the China Scholarship Council.

\appendix

\section{A conformal representation of Minkowski in Stewart's gauge}
\label{Appendix:ConformalMinkowskiStewartGauge}

In this appendix we discuss a conformal representation of the
Minkowski spacetime in Stewart's gauge. 

\smallskip
In the following,  let $(\tilde{\mathcal{M}},\tilde{\bmeta})$ denote
the Minkowski spacetime and let $\bar{x}=(x^\mu)$ correspond to
standard Cartesian coordinates so that
\begin{align*}
& \tilde{\bmeta} = \eta_{\mu\nu} \mathbf{d}x^\mu\otimes
   \mathbf{d}x^\nu\\
& \phantom{ \tilde{\bmeta}}= \mathbf{d} x^0\otimes \mathbf{d}x^0 -
   \mathbf{d}x^1\otimes\mathbf{d}x^1 -\mathbf{d}x^2\otimes
   \mathbf{d}x^2 -\mathbf{d}x^3\otimes \mathbf{d}x^3.
\end{align*}
Consider now the domain 
\begin{align*}
\tilde{\mathcal{D}} \equiv \{ p\in \mathbb{R}^4 \;|\; \eta_{\mu\nu}
x^\mu(p) x^\nu{p}<0 \}, 
\end{align*}
corresponding to the complement of the light cone through the
origin. A suitable conformal representation of this domain is obtained
via the coordinate inversion defined by
\begin{align*}
y^\mu = - \frac{x^\mu}{X^2}, \qquad X^2\equiv \eta_{\mu\nu} x^\mu x^\nu,
\end{align*}
so that 
\begin{align*}
\eta_{\mu\nu} \mathbf{d}y^\mu\otimes\mathbf{d}y^\nu
= X^{-4} \eta_{\mu\nu} \mathbf{d}x^\mu\otimes\mathbf{d}x^\nu.
\end{align*}
Accordingly defining the conformal factor $\Xi \equiv X^{-2}$ one
obtains the conformal metric $\bmeta = \Xi^2 \tilde{\bmeta}$. Observe
that
\begin{align*}
\bmeta =  \eta_{\mu\nu} \mathbf{d}y^\mu\otimes
   \mathbf{d}y^\nu,
\end{align*}
so that this conformal representation of~$\tilde{\mathcal{D}}$ is flat
---in particular, the Ricci scalar of~$\bmeta$ vanishes.  Of
particular relevance for the subsequent analysis is that future null
infinity is described by
\begin{align*}
\mathscr{I}^+ =\{ p\in\mathbb{R}^4 \;|\; y^0(p)>0,\; \eta_{\mu\nu}
y^\mu(p)y^\nu(p)=0 \}. 
\end{align*}
In order to set up Stewart's gauge consider, first, standard spherical
coordinates~$(t,t,\theta,\varphi)$ and then, in turn, double null
coordinates~$(u,v,\theta,\phi)$ such that
\begin{align*}
u =\frac{1}{\sqrt{2}}(t-r), \qquad v=\frac{1}{\sqrt{2}}(t+r),
\end{align*}
for which one has
\begin{align}
\bmeta = \mathbf{d}u\otimes\mathbf{d}v + \mathbf{d}v\otimes\mathbf{d}u
-(u-v)^2\bmsigma.
\label{MinkowskiStewart}
\end{align}
In particular, $\mathscr{I}^+$ is given by the
condition~$v=0$. Redefining~$v$ through an translation one can set the
location of~$\mathscr{I}^+$ to be given by the
condition~$v=v_\bullet$, for some~$v_\bullet>0$. A NP frame satisfying
the conditions of Stewart's gauge can be obtained by setting
\begin{align*}
\bml = \bmpartial_v, \qquad \bmn = \bmpartial_u, \qquad \bmm =
\frac{(u-v)}{\sqrt{2}}(\bmpartial_\theta + \mbox{i} \csc\theta
\bmpartial_\bmvarphi),
\end{align*}
so that 
\begin{align*}
Q=1, \qquad C^{\mathcal{A}} =0, \qquad P^2= \frac{u-v}{\sqrt{2}},
\qquad P^3= \frac{\mbox{i} (u-v)}{\sqrt{2}}\csc\theta.
\end{align*}
A computation readily shows that the NP spin connection coefficients
associated to the above tetrad are
\begin{align*}
&\kappa=\tau =\sigma=\pi=\nu=\lambda=\epsilon =\gamma=0, \\
  & \rho=\mu =\displaystyle \frac{1}{u-v}, \quad \alpha =-\beta
  = \displaystyle\frac{\cot \theta }{2\sqrt{2}(u-v)} .  
\end{align*}
The above coefficients can be readily seen to satisfy the conditions
of Stewart's gauge. Moreover, as the metric \eqref{MinkowskiStewart}
is flat one has that
\begin{align*}
& \phi_0 =\phi_1=\phi_2=\phi_3=\phi_4=0, \\
&
  \Phi_{00}=\Phi_{01}=\Phi_{02}=\Phi_{10}=\Phi_{11}=\Phi_{12}=\Phi_{20}=\Phi_{21}=\Phi_{22}=0,
                                        \\
&\Lambda =0.
\end{align*}
The above expressions imply regular \emph{characteristic data} as
given by Lemma~\ref{Lemma3}.



\end{document}